\documentclass[aps,floatfix,twocolumn,footinbib]{revtex4} \usepackage{epsfig}

\usepackage{mdframed}
\usepackage{amsthm}
\usepackage{amssymb}
\usepackage{graphicx} 
\usepackage{hyperref} 
\usepackage{tikz}
\usepackage{braket}
\usepackage{dsfont}
\usepackage{mathtools}
\usepackage{mathrsfs}
\newtheorem{remark}{Remark}
\newtheorem{exmp}{Example}
\newtheorem{Lemma}{Lemma}
\newtheorem{defn}{Definition}
\newtheorem{thm}{Theorem}
\newtheorem{Corollary}{Corollary}

\begin{document}
\title{Identification of  single-input-single-output quantum linear systems}
\author{Matthew Levitt}
\email{pmxml2@nottingham.ac.uk}
\affiliation{School of Mathematical Sciences, University of Nottingham, University Park, NG7 2RD Nottingham, United Kingdom}
\author{M\u{a}d\u{a}lin Gu\c{t}\u{a}}
\affiliation{School of Mathematical Sciences, University of Nottingham, University Park, NG7 2RD Nottingham, United Kingdom}

\begin{abstract}

The purpose of this paper is to investigate system identification for single-input-single-output general (active or passive) quantum linear systems. For a given input we address the following questions: (1) Which parameters can be identified by measuring the output? (2) How can we construct a system realization from sufficient input-output data? 

We show that for time-dependent inputs, the systems which cannot be distinguished are related by symplectic transformations acting on the space of system modes. This complements a previous result of \cite{Guta2} for passive linear systems. In the regime of stationary quantum noise input, the output is completely determined by the power spectrum. We define the notion of global minimality for a given power spectrum, and characterize globally minimal systems as those with a fully mixed stationary state. We show that in the case of systems with a cascade realization, the power spectrum completely fixes the transfer function, so the system can be identified up to a symplectic transformation. We give a method for constructing a globally minimal subsystem direct from the power spectrum. Restricting to passive systems the analysis simplifies so that identifiability may be completely understood from the eigenvalues of a particular system matrix.

\end{abstract}\maketitle

\section{Introduction}

We are currently witnessing the beginning of a quantum technological revolution aimed at harnessing features that are unique to the  quantum world such as coherence, entanglement and uncertainty, for practical applications in metrology, computation, information transmission and cryptography \cite{Nielsen,dowling}.
%
The high sensitivity and limited controllability of quantum dynamics has stimulated the development of theoretical and experimental techniques at the overlap between quantum physics and ``classical'' control engineering, such as quantum filtering \cite{MILB, LB}, feedback control \cite{FEED1,FEED2, DOHERTY, Yanagisawa}, network theory \cite{g2, g3, NETWORK, nurdin}, and linear systems theory \cite{petersen, nurdin, peter, JNP, squeezing, Guta2, cascade, cascade2, new, Indep, Levitt}. 

%
%
%

In particular, there has been a rapid growth in the study of quantum linear systems (QLSs), with many applications, e.g., quantum optics, opto-mechanical systems, quantum memories, entanglement generation, electrodynamical  systems and cavity QED systems   \cite{Naoki, wall, Tian, GZ, MILB, STOCK, memory, memory2, ZHANG1, MATYAS, DOHERTY}.


 \begin{figure}[h]
\centering
\includegraphics[scale=0.30]{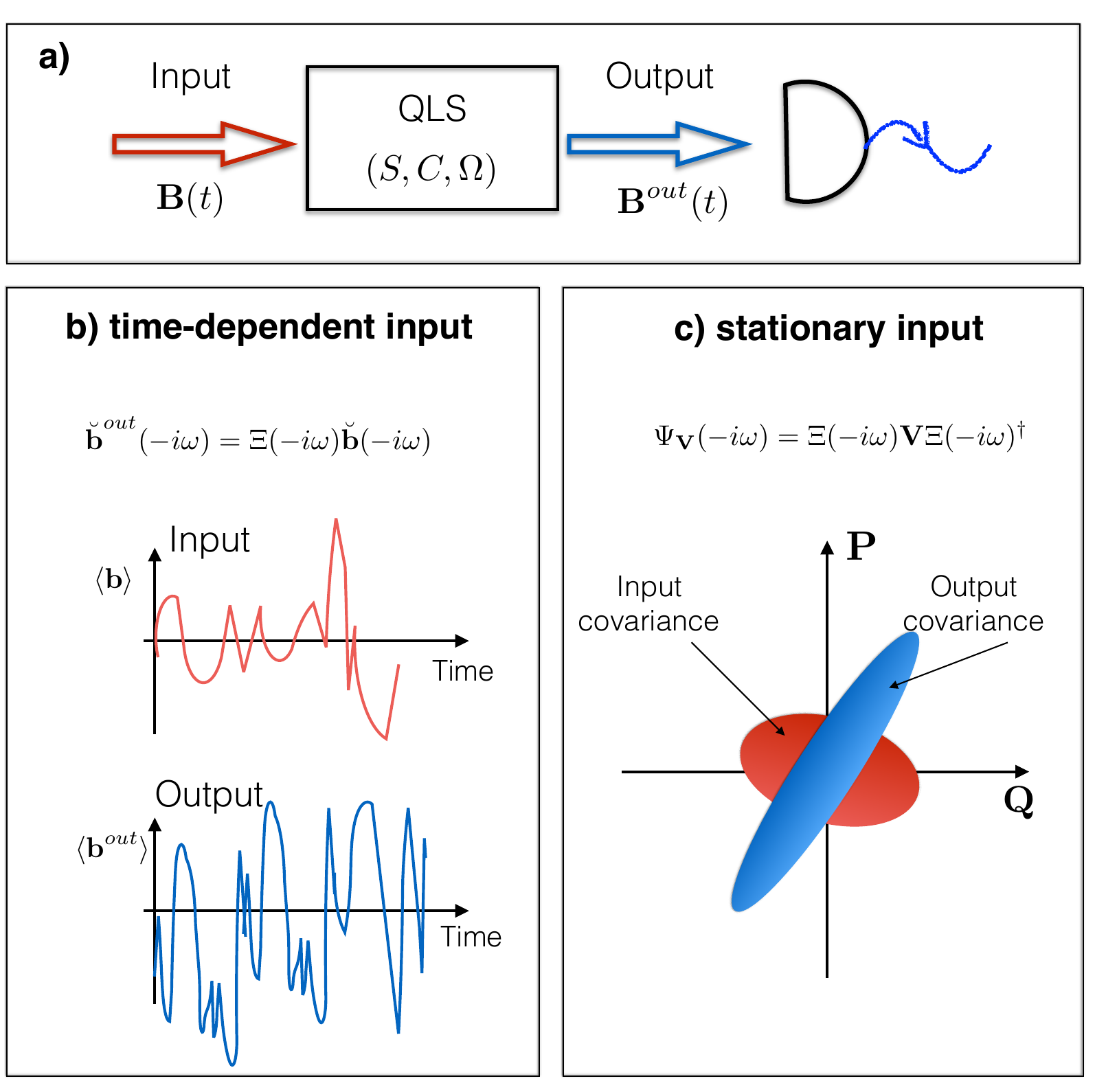}
\caption{(a) System identification problem: find parameters $(S,C,\Omega)$ of a linear input-output system by measuring output. (b) Time-dependent scenario: in frequency domain, input and output are related by the {\it transfer function} $\Xi(-i\omega)$ which depends on $(S,C,\Omega)$.  (c) Stationary scenario: \textit{power spectrum} describes output covariance which is quadratic with respect to 
$\Xi(-i\omega)$ .\label{twoapproach}}
\end{figure}

System identification theory \cite{Ljung,Ljung2,Pintelon&Schoukens, GUTA3, GUTA4} lies at the interface between control theory and statistical inference, and deals with the estimation of unknown parameters of dynamical systems and processes from input-output data. 
The integration of control and identification techniques plays an important role, e.g., in adaptive control \cite{Astrom}.  
The identification of linear systems is by now a well developed subject in classical systems theory \cite{glover, kalman, Ljung, Ljung2, Ho, anders, Youla, zhou, Pintelon&Schoukens, Davies, nerve}, but has not been fully explored in the quantum domain \cite{Guta2}.

This paper deals with the problem of identifying unknown dynamical parameters of \textit{ quantum linear systems} (QLSs). A QLS is a continuous variables open system with modes ${\bf a}= (a_1,\dots, a_n)^T$, which has a quadratic Hamiltonian, and couples linearly to Bosonic input channels ${\bf B}(t)= (B_1(t), \dots , B_m(t))^T$ representing the environmental degrees of freedom in the time domain. The system and environment modes satisfy the commutation relations 
$$
[{\bf a}, {\bf a}^\dagger]= \mathds{1}_n, \quad [{\bf b}(t), {\bf b}(s)^\dagger] = \delta(t-s)\mathds{1}_m, 
$$
where ${\bf b}(t)= \frac{d{\bf B}(t)}{dt}$ is the infinitesimal annihilation operator at time $t$.
The joint dynamics is completely characterized by the triple $(S, C, \Omega)$ consisting of a $2m\times 2m$ scattering matrix $S$, a $2m\times 2n$ system-input coupling matrix $C$, and a $2n \times 2n$ Hamiltonian matrix 
$\Omega$. 
Since each system or channel mode has two coordinates corresponding to creation and annihilation operators, all matrices have  a $2\times 2$ block structure, and it is convenient to use the ``doubled-up'' conventions introduced in \cite{squeezing}, as detailed in Sec. \ref{red1}. The data $(S, C, \Omega)$  fix the joint unitary dynamics ${\bf U}(t)$ obtained as a solution of a quantum stochastic differential equation \cite{path}; due to the quadratic interactions, the evolved modes ${\bf a}(t) := {\bf U}(t)^{\dag} {\bf a} {\bf U}(t)$ and output fields ${\bf B}^{out}(t) := {\bf U}(t)^{\dag} {\bf B}(t) {\bf U}(t)$ are linear transformations of the original degrees of freedom.

In a nutshell, system identification deals with the estimation of dynamical parameters of input-output systems from data obtained by performing measurements on the output fields. 
We distinguish two contrasting approaches to the identification of linear systems, 
which we illustrate in Fig. \ref{twoapproach}. 
In the first approach, one probes the system with a known \emph{time-dependent} input signal (e.g., coherent state), then uses the output measurement data to compute an estimator of the unknown dynamical parameter. In the Laplace domain, the input and output fields are related by a linear transformation given by the $2m\times 2m$ \emph{transfer function} $\Xi(s)$:
\begin{equation}\label{eq.input-output}
\breve{\bf b}^{out}(s) = \Xi(s) \breve{\bf b}(s), \quad 
\end{equation}
where $\breve{\bf b}(s)$ is the vector of input creation and annihilation input noise operators. 
The transfer function $\Xi(s)$ is a rational matrix valued function, which becomes a symplectic matrix in the 
``frequency domain'' (i.e., for $s=-i\omega\in i\mathbb{R}$), reflecting the fact that the unitary dynamics preserves the canonical commutation relations. Similarly to the classical case, Eq. \eqref{eq.input-output} means that the input-output data can be used to reconstruct the transfer function $\Xi(s)$, while systems with the same transfer function cannot be distinguished. Therefore, the basic identifiability problem is to find the equivalence classes of systems with the same transfer function. 

In \cite{Guta2} this problem was analyzed for the special class of \emph{passive} quantum linear systems (PQLSs) and it was shown that \emph{minimal} equivalent systems are related by $n\times n$ \emph{unitary} transformations acting on the space of annihilation modes ${\bf a}$. 
By definition a QLS is minimal if no lower dimensional system has the same transfer function, which in the passive case is equivalent to the system being either observable, controllable, or Hurwitz stable \cite{Guta2}. In Sec. \ref{identifiability} we answer the identifiability question for the case of general (not necessarily passive) QLSs; we show that the equivalence classes are determined by \emph{symplectic} transformations acting on the doubled-up space of canonical variables $\breve{\bf a}$. It is worth noting that while in the classical set-up equivalent linear systems are related by \emph{similarity} transformations, 
in both  quantum scenarios described above the transformations are more restrictive due to the unitary nature of the dynamics. 
 
In the second approach, the input fields are prepared in a \emph{stationary in time}, pure Gaussian state 
with independent increments (squeezed vacuum noise), which is completely characterised by the covariance matrix ${V}={V}(N,M)$ and the associated quantum Ito rule \cite{squeezing} 
\begin{align*}\label{eq:Ito1}
\left(\begin{smallmatrix}{d}\mathbf{B}(t){d}\mathbf{B}(t)^{\dag}&{d}\mathbf{B}(t){d}\mathbf{B}(t)^T\\{d}\mathbf{B}^{\#}(t){d}\mathbf{B}(t)^{\dag}&{d}\mathbf{B}^{\#}(t){d}\mathbf{B}(t)^T\end{smallmatrix}\right)
=\left(\begin{smallmatrix} N^T+\mathds{1}&M\\M^{\dag}&N\end{smallmatrix}\right){d}t:={V}t.
\end{align*}
If the system is minimal and Hurwitz stable, the dynamics exhibits an initial transience period after which it reaches stationarity and the output is in a stationary Gaussian state, whose covariance in the frequency domain is  given by the \emph{power spectrum} 
$$ 
\Psi_{V}(-i\omega) = \Xi(-i\omega) {V} \Xi(-i\omega)^\dagger.
$$
Since the power spectrum depends quadratically on the transfer function, the parameters which are identifiable in the  stationary scenario will also be identifiable in the time-dependent one. Our goal is to understand to what extent the converse is also true. First, we note that for a given minimal system there may exist lower dimensional systems with the same power spectrum. To understand this, consider the system's stationary state and note 
that it can be uniquely written as a tensor product between a pure and a mixed Gaussian state 
(cf. the symplectic decomposition). In Theorem \ref{equivalence} we show that  restricting the system to the mixed component leaves the power spectrum unchanged. Furthermore,  the pure component is passive, which ties in with previous results of \cite{Naoki}. Conversely, if the stationary state is fully mixed, there exists no smaller dimensional system with the same power spectrum. Such systems will be called \emph{globally minimal}, and can be seen as the analog of minimal systems for the stationary setting.

One of the main results is Theorem \ref{mainresult} which shows that for ``generic'' globally minimal single-input-single-output (SISO) systems which admit a cascade representation, the power spectrum $\Psi_{V}(s)$ determines the transfer function $\Xi(s)$ uniquely, and therefore the time-dependent and time-stationary identifiability problems are equivalent.  It is interesting to note that this equivalence is a consequence of unitarity and purity of the input state, and does not hold for generic classical linear systems \cite{glover, anders}.

The paper is  structured as follows. In Sec. \ref{red1} we review the setup of input-output QLSs, and their associated transfer function. We discuss in greater detail the two identifiability approaches mentioned above. In Sec. \ref{identifiability} we study the identifiability of QLSs in the time-dependent input setting. In Theorem \ref{symplecticequivalence} we show that the equivalence classes of input-output systems with the same transfer function are given by symplectic transformations of the system's modes. We further show how a physical realization can be constructed from the system's transfer function.
In Sec. \ref{powerspectrum} we analyze the identifiability of QLSs in a stationary Gaussian noise input setting. 
We introduce the notion of global minimality for systems with minimal dimension for a given power spectrum, and show that a system is globally minimal if and only if it has a fully mixed stationary state,  cf. Theorem \ref{equivalence}. In Sec. \ref{main1} we analyze the structure of the power spectrum identifiability classes, and show that the power spectrum determines the transfer function uniquely, for a large class of SISO systems, cf. Theorem \ref{mainresult}. Finally, we show that  using an additional input channel with an appropriately chosen entangled input ensures that the system is always globally minimal.

\subsection{Preliminaries and notation}

We use the following notations: 
``Tr'' and ``Det'' denotes the trace and determinant 
of a matrix, respectively. 
For a matrix $X=(X_{ij})$ the  symbols:
$X^{\#}=(X_{ij}^*)$,  $X^{T}=(X_{ji})$, $X^{\dag}=(X_{ji}^*)$
represent the complex conjugation, transpose, and adjoint matrix respectively, where ``*'' indicates complex conjugation. We also use the doubled-up notation $\breve{X}:=\left[X^T, (X^{\#})^T\right]^T$ and $\Delta(A, B):=\left[A, B; B^{\#}, A^{\#}\right]$. For example, we may write the transformation $Y=AX+BX^{\#}$ in \textit{doubled-up} form as $\breve{Y}=\Delta(A,B)\breve{X}$. For a matrix $Z\in\mathbb{R}^{2n\times2m}$ define $Z^\flat=J_mZ^{\dag}J_n$, where $J_n=\left[\mathds{1}_n, 0; 0, -\mathds{1}_n\right]$. $\mathrm{Spec}(X)$ is the set of all distinct eigenvalues of $X$. A similar notation is used for matrices of operators. 
We use ``$\mathds{1}$'' to represent the identity matrix or operator. $\delta_{jk}$ is Kronecker $\delta$ and $\delta(t)$ is Dirac $\delta$.
The commutator is denoted by $[\cdot, \cdot]$.

\begin{defn}\label{def.symplectic}
A matrix $S\in\mathbb{C}^{2m\times 2m}$ is said to be \textit{$\flat$- unitary} if it is invertible and satisfies
\[S^\flat S=SS^\flat =\mathds{1}_{2m}.\]

If additionally, $S$ is of the form $S=\Delta(S_-, S_+)$ for some $S_-, S_+\in\mathbb{R}^{m\times m}$ then we say that it is \textit{symplectic}. Such matrices form a group called the \textit{symplectic group} \cite{squeezing, ultracoherence}.
\end{defn}


\section{Quantum Linear systems}\label{red1}

In this section we briefly review the QLS theory, highlighting along 
the way results that will be relevant for this paper. 
We refer to \cite{GZ} for a more detailed discussion on the input-output 
formalism, and  to the review papers \cite{peter, path, path2, nurdin} 
for the theory of linear systems.

\subsection{Time-domain representation}\label{timed}

A linear input-output quantum system is defined as a continuous 
variables (cv) system coupled to a Bosonic environment, such that their joint 
evolution is linear in all canonical variables. 
The system is described by the column vector of annihilation operators, 
$\mathbf{a}:=[\mathbf{a}_1,\mathbf{a}_2, \dots, \mathbf{a}_n]^T$, 
representing the $n$ cv modes. 
Together with their respective creation operators 
$\mathbf{a}^{\#}:=[\mathbf{a}^{\#}_1,\mathbf{a}^{\#}_2, \dots, \mathbf{a}^{\#}_n]^T$ 
they satisfy the canonical commutation relations (CCR)
$
\left[\mathbf{a}_i, \mathbf{a}^{*}_j\right]=\delta_{i j}\mathds{1}.
$
We denote by $\mathcal{H}:= L^2(\mathbb{R}^n)$ the Hilbert space of the system carrying the standard representation of the 
$n$ modes. The environment is modelled by $m$ bosonic fields, called \textit{input channels}, 
whose fundamental variables are the fields 
$\mathbf{B}(t):=\left[\mathbf{B}_{1}(t), \mathbf{B}_{2}(t), \ldots, 
\mathbf{B}_{m}(t)\right]^T$, where $t\in \mathbb{R}$ represents time. 
The fields satisfy the CCR 
\begin{eqnarray}
\left[\mathbf{B}_{i}(t), \mathbf{B}^{\#}_{j}(s)\right]=\mathrm{min}\{t,s\}\delta_{ij}\mathds{1}.
\end{eqnarray}
Equivalently, this can be written as 
$\left[\mathbf{b}_{i}(t),\mathbf{b}_{j}^{\#}(s)\right]
=\delta(t-s)\delta_{ij}\mathds{1}$, where 
$\mathbf{b}_{i}(t)$ are the infinitesimal (white noise) annihilation operators formally defined as 
$\mathbf{b}_{i}(t):={d}\mathbf{B}_{i}(t)/{d}t$ \cite{peter}. 
The operators can be defined in a standard fashion on the Fock space 
$\mathcal{F}= \mathcal{F}(L^2(\mathbb{R})\otimes \mathbb{C}^m)$ \cite{LB}. For most of  the paper we consider the scenario where the input is prepared in a \emph{pure, stationary in time, mean-zero, Gaussian state} with independent increments characterized by the covariance matrix 
\begin{align}\label{Ito}
\nonumber\left<\begin{smallmatrix}\mathbf{B}(t)\mathbf{B}(t)^{\dag}&{d}\mathbf{B}(t){d}\mathbf{B}(t)^T\\{d}\mathbf{B}^{\#}(t){d}\mathbf{B}(t)^{\dag}&{d}\mathbf{B}^{\#}(t){d}\mathbf{B}(t)^T\end{smallmatrix}\right>&=\left(\begin{smallmatrix} N^T+\mathds{1}&M\\M^{\dag}&N\end{smallmatrix}\right){d}t
\\&:={V}(N,M){d}t,
\end{align}
where the brackets denote a quantum expectation. 
Note that $N=N^{\dag}$, $M=M^T$, and $V\geq0$,   which ensures that the state does not violate the uncertainty principle. The state's purity can be characterized in terms of the symplectic eigenvalues of $V$, as will be discussed in Sec. \ref{powerspectrum}. In particular, $N=M=0$ corresponds to the  vacuum state, while pure squeezed states for  single-input-single-output (SISO) systems (i.e., $m=1$)  satisfy $|M|^2= N(N+1)$. More generally, we consider a nonstationary scenario where the input state has time-dependent mean $\langle \mathbf{B}(t)\rangle$, e.g., a coherent state with time-dependent amplitude. For more details on Gaussian states see \cite{adduce, weedbrook}.


The dynamics of a general input-output system is determined by the system's Hamiltonian and its coupling to the environment. In the Markov approximation, the joint unitary evolution of system and environment is described by the (interaction picture) unitary ${\bf U}(t)$ on the joint space 
$\mathcal{H}\otimes \mathcal{F}$, which is the solution of the quantum stochastic differential equation  \cite{LB,Dong,GZ,path,path2}
\begin{eqnarray}
\label{eq.QSDE} 
&&{d}\mathbf{U}(t) 
         :=\mathbf{U}(t+{d}t)-\mathbf{U}(t) \\
          &&=\left(-i \mathbf{H}{d}t             + \mathbf{L}{d}\mathbf{B}(t)^{\dag}
               - \mathbf{L}^{\dag}{d}\mathbf{B}(t) 
               -\frac{1}{2}\mathbf{L}^{\dag}\mathbf{L}dt\right)\mathbf{U}(t),\nonumber
\end{eqnarray}               
%
with initial condition ${\bf U}(0)= \mathbf{I}$. 
Here, ${\bf H}$ and ${\bf L}$ are system operators describing the 
system Hamiltonian and coupling to the fields; 
${d}{\bf B}_i(t), {d}{\bf B}_i^{\#}(t)$, 
are increments of fundamental quantum stochastic processes describing the 
creation and annihilation operators in the  input channels.  

For the special case of \emph{linear} systems, 
the coupling and Hamiltonian operators are of the form  
\begin{eqnarray*}
\mathbf{L} &=& C_{-}\mathbf{a}+C_{+}\mathbf{a}^{\#},\\
\mathbf{H}&=&  \mathbf{a}^{\dag}\Omega_{-}\mathbf{a}+\frac{1}{2}\mathbf{a}^T\Omega_{+}^{\dag}\mathbf{a}+\frac{1}{2}\mathbf{a}^{\dag}\Omega_{+}\mathbf{a}^{\#},
\end{eqnarray*}
for $m\times n$ matrices $C_{-}, C_{+}$ and $n\times n$ matrices $\Omega_{-}, \Omega_{+}$ satisfying $\Omega_{-}=\Omega_{-}^{\dag}$ and $\Omega_{+}=\Omega_{+}^T$. 

As shown below, this ensures that all canonical variables evolve linearly in time. Indeed, let $\mathbf{a}(t)$ and 
$\mathbf{B}^{out}(t)$ be the Heisenberg evolved system and output variables 
\begin{eqnarray}
       \mathbf{a}(t):=\mathbf{U}(t)^{\dag}\mathbf{a}\mathbf{U}(t),~~~ 
       \mathbf{B}^{out}(t):=\mathbf{U}(t)^{\dag}\mathbf{B}(t)\mathbf{U}(t).
\end{eqnarray}
By using the QSDE \eqref{eq.QSDE} and the Ito rules (\ref{Ito}) one can obtain the 
following Ito-form quantum stochastic differential equation of the QLS in the doubled-up notation \cite{squeezing}
\begin{eqnarray}\label{langevin}
      {d}\breve{\bf a}(t) &=& 
            A \breve{\bf a}(t){d}t-C^{\flat}{d} \breve{\bf B}(t),\\
       {d} \breve{\bf B}^{out}(t) &=& 
            C \breve{\bf a}(t){d}t+{d} \breve{\bf B}(t),
\end{eqnarray}
where $ \breve{\bf a} := ({\bf a}^T, {\bf a^{\#}}^T)^T$,  $C:=\Delta\left(C_{-}, C_{+}\right)$, and $A:=\Delta\left(A_{-}, A_{+}\right)=-\frac{1}{2}C^\flat C-iJ_n\Omega$ with $\Omega= \Delta\left(\Omega_{-}, \Omega_{+}\right)$ and
\[
A_{\mp}:=-\frac{1}{2}\left(C_{-}^{\dag}C_{\mp}-C_{+}^{T}C_{\pm}^{\#}\right)-i\Omega_{\mp}.
\] 
It is important to note that not all choices of  $A$ and $C$ may be physically realizable as open quantum systems \cite{JNP}.

A special case of linear systems is that  of \textit{passive} quantum linear systems (PQLSs) for which $C_{+}=0$ and 
$\Omega_{+}=0$, whose system identification theory was studied in \cite{Guta2}. We will return to this important class along the way.
This type of system often arises in applications, and includes optical cavities and beam splitters.

\subsection{Controllability and observability}
By taking the expectation with respect to the initial joint system state of Eqs. (\ref{langevin}) we obtain the following classical linear system
\begin{eqnarray}\label{classicallangevin}
       {d}\left<\breve{\bf a}(t) \right>= 
            A\left< \breve{\bf a}(t) \right>{d}t-C^{\flat}{d}\left<\breve{\bf B}(t)\right>,\\
       {d}\left<\breve{\bf B}^{out}(t)\right> = 
            C\left<\breve{\bf a}(t) \right>{d}t+{d}\left<\breve{\bf B}(t)\right>.
\end{eqnarray}

\begin{defn}
The quantum linear system (\ref{langevin}) is said to be Hurwitz stable (respectively controllable, observable) if the corresponding classical system (\ref{classicallangevin}) is Hurwitz stable (respectively controllable, observable).
\end{defn}
In general, for a quantum linear system observability and controllability are equivalent \cite{Indep}. A system possessing one (and hence both) of these properties is called \textit{minimal}. Checking minimality comes down to verifying that the rank of the following observability matrix is $2n$:
\[\mathcal{O}=[C^T, (CJ_n\Omega)^T,\hdots, \left(C(J_n\Omega)^{2n-1}\right)^T]^T,\]
where $\Omega=\Delta(\Omega_{-},  \Omega_{+})$.
In the case of passive systems Hurwitz stability is further equivalent to minimality of the system \cite{Guta2}. However for active systems, although the statement [Hurwitz $\implies$ minimal]  is true \cite{Naoki}, the converse statement ([minimal $\implies$ Hurwitz]) is not necessarily so. We see this by means of a counterexample.   
\begin{exmp}
Consider a general one-mode SISO QLS, which is parametrizsed by $\Omega=\Delta(\omega_{-},  \omega_{+})$ and $C=\Delta\left(c_-, c_+\right)$.
The system is Hurwitz stable (i.e. the eigenvalues of $A$ have a strictly negative real part) if and only if
\begin{enumerate}
\item[(1)] $|c_-|>|c_+|$ and $|\omega_-|\geq|\omega_+|$, or
\item[(2)] $|\omega_+|>|\omega_-|$  and $\sqrt{|\omega_+|^2-|\omega_-|^2}<\frac{1}{2}\left(|c_-|^2-|c_+|^2\right)$.
\end{enumerate}
A system is nonminimal if and only if the following matrix has rank less than 2: 
\[
\left[
\begin{smallmatrix}
C\\CJ_n\Omega
\end{smallmatrix}
\right]=
\left[\begin{array}{cc}
c_-&c_+
\\
{c_+}^{\#}&{c_-}^{\#}
\\
c_-\omega_- -c_+{\omega_+}^{\#} &c_-\omega_+ -c_+\omega_-
\\ 
{c_+}^{\#}\omega_- -{c_-}^{\#}{\omega_+}^{\#}   &  {c_+}^{\#}\omega_+ -{c_-}^{\#}\omega_-
\end{array}\right]. 
\]

Clearly it is possible for a system to be \{minimal\}$\cap$\{Hurwitz\} or \{non-minimal\}$\cap$\{non-Hurwitz\}. Further, for a  counterexample to the statement:  [minimal $\implies$ Hurwitz] consider for example   $|c_+|>|c_-|$ with $\omega_+=\omega_-$.
\end{exmp}

In light of the previous example, we make the physical assumption that all systems considered throughout this paper are Hurwitz (hence minimal).

\subsection{Frequency-domain representation}\label{freqrep}

For linear systems it is often useful to switch from the time domain dynamics 
described above, to the frequency domain picture. 
Recall that the Laplace transform of a generic process ${\bf x}(t)$ is defined by
\begin{equation}\label{eq.laplace.classic}
 \mathbf{x}(s) := \mathcal{L}[\mathbf{x}](s)
      =\int_{-\infty}^\infty e^{-st}{\bf x}(t)dt, 
\end{equation}
where $s\in\mathbb{C}$. 
In the Laplace domain the input and output fields are related as follows 
\cite{Yanagisawa}:
\begin{equation}
\breve{\bf b}^{out}(s) = 
\Xi(s)
 \breve{\bf b}(s),
\label{iol}
\end{equation}
where $ \Xi(s)$ is the  \emph{transfer function matrix} of the system%
\begin{equation}
\label{eq.transfer.function.general}
       \Xi(s)=\Big\{\mathds{1}_m-C(s\mathds{1}_n-A)^{-1}C^{\flat}\Big\}=
       \left(\begin{smallmatrix}
       \Xi_{-}(s)&\Xi_{+}(s)
       \\
       \Xi_{+}(s^{\#})^{\#}& \Xi_{-}(s^{\#})^{\#}
       \end{smallmatrix}\right).
\end{equation}

In particular, the frequency domain input-output 
relation is
$\breve{\bf b}^{out}(-i\omega) = 
\Xi(-i\omega)
 \breve{\bf b}(-i\omega).$
The corresponding commutation relations are
$\left[\mathbf{b}(-i\omega),\mathbf{b}(-i\omega')^{\#}\right]
=i\delta(\omega-\omega')\mathds{1}$, and similarly for the output modes \footnote{Note that the position of the conjugation sign is important here because in general  $\mathbf{b}(-i\omega')^{\#}$ and $\mathbf{b}^{\#}(-i\omega')$ are not the same, cf. Definition (\ref{eq.laplace.classic}).}. 
As a consequence, the transfer matrix $\Xi(-i\omega)$ is symplectic for all frequencies $\omega$ \cite{squeezing}. 

More generally one may allow for static scattering (implemented by passive optical components such as beamsplitters) or static squeezing processes to act on the interacting field before interacting with the system. The corresponding transfer function is obtained by multiplying the transfer function \eqref{eq.transfer.function.general} with the scattering or squeezing symplectic matrix $S$ on the right \cite{squeezing}. 

In the case of passive systems,  $\Xi_{+}(s)\equiv0$ and so  the doubled-up notation is no longer necessary; the input-output relation becomes \cite{Yanagisawa,Guta2}
\begin{equation}
\label{iolp}
\mathbf{ b}^{out}(s) = 
\Xi(s)
 \mathbf{ b} (s),
\end{equation}
where the transfer function is given by
\begin{equation}
\label{eq.transfer.function}
       \Xi(s)=\Big\{\mathds{1}_m-C_{-}(s\mathds{1}_n-A_{-})^{-1}C_{-}^{\dag}\Big\}S,
\end{equation}
which is unitary for all $s= -i\omega\in i \mathbb{R}$.
In the case of passive systems we write the triple determining the evolution as 
$\left(S, C_-, \Omega_{-}\right)$, where the scattering matrix $S$ is unitary.

Finally, we note that while the transfer function is uniquely determined by the triple $(S, C, \Omega)$, the converse  statement is not true, as discussed in detail in the next section.

\section{Transfer function identifiability}\label{identifiability}


\subsection{Identifiability classes}

We now consider the following general question: which dynamical parameters 
of a QLS can be identified by observing the output fields for appropriately 
chosen input states? 
This is the quantum analog of the classical system identification problem addressed in \cite{kalman, Ho, anders}. 
The input-output relation \eqref{iol} shows that the experimenter can at most  
identify the transfer function $\Xi(s)$ of the system. 
Systems which have the same transfer function are called {\it equivalent} and belong to the same \emph{equivalence class}.

Before answering this question for general QLSs we discuss the case of passive QLSs considered in 
\cite{Guta2}.
The transfer function in Eq. \eqref{iolp} can be identified by sending a coherent input signal of a given frequency 
$\omega$ and known amplitude $\alpha(\omega)$, and measuring the 
output state, which is a coherent state of the same frequency and amplitude 
$\Xi (-i\omega) \alpha(\omega)$. 

In the case of passive systems it is known that
two minimal systems with parameters $(\Omega,C,S)$ and $(\Omega', C',S')$ are 
equivalent if and only if their parameters are related by a unitary transformation, 
i.e. $C'=CT$ and $\Omega'=T\Omega T^{\dag}$ for some $n\times n$ unitary 
matrix $T$, and $S= S'$. 
The first part of this result was shown in \cite{Guta2}; the fact that the scattering 
matrices must be equal follows by choosing $s=-i\omega$ and taking the limit 
$\omega\to \infty$ in Eq. (\ref{eq.transfer.function}). 
Physically, this means that at frequencies far from the internal frequencies 
of the system, the input-output is dominated by the scattering or squeezing between the 
input fields. 
Our first main result is to extend this result to general (active) linear systems.

\begin{thm}\label{symplecticequivalence}
Let $\left({S}, C, \Omega\right)$ and $\left({S'}, C', \Omega'\right)$ be two minimal, and stable QLSs. 
Then they have the same transfer function if and only if there exists a symplectic matrix $T$ such that 
\begin{equation}\label{eq.equivalene.classes}
J_n\Omega'=TJ_n\Omega T^{\flat}, \,\,\, C'=CT^{\flat}\,\,\,S=S'.
\end{equation}
\end{thm}
\begin{proof}
Firstly, using the same  argument as above, the scattering or squeezing matrices $S$ and $S'$ must be equal.

It is known \cite{Ljung} that two minimal classical linear systems
\[d\mathbf{x}(t)=A\mathbf{x}(t)dt+B\mathbf{u}(t)dt,\,\,\,d\mathbf{y}(t)=C\mathbf{x}(t)dt+D\mathbf{u}(t)dt\]
and 
\[d\mathbf{x}(t)=A'\mathbf{x}(t)dt+B'\mathbf{u}(t)dt,\,\,\,d\mathbf{y}(t)=C'\mathbf{x}(t)dt+D'\mathbf{u}(t)dt\]
for input $\mathbf{u}(t)$, output $\mathbf{y}(t)$, and system state $\mathbf{x}(t)$ have the same transfer function if and only if
\[A'=TAT^{-1},\,\,\, B'=TB,\,\,\,C'=CT^{-1},\,\,\,D'=D\]
for some invertible matrix $T$.
Hence, for our setup $C\left(s\mathds{1}-A\right)^{-1}C^{\flat}=C'\left(s\mathds{1}-A'\right)^{-1}C'^{\flat}$
if and only if there exists an invertible matrix $T$ such that 
  \[A'=TAT^{-1},\,\,\, C'^{\flat}=TC^{\flat},\,\,\,C'=CT^{-1}.\]
Note that at this stage $T$ is not assumed to be symplectic.
The second and third conditions imply $C=C\left(T^\flat T\right)$, which further implies that $\lbrack T^\flat T,C^{\flat}C\rbrack=0$. Now by earlier definitions $A=-\frac{1}{2}C^\flat C-iJ_n\Omega$, so that the second and third conditions applied to the first condition imply that $J_n\Omega'=TJ_n\Omega T^{-1}$. Next, using this and the  observation $\left(J_n\Omega\right)^\flat =J_n\Omega$ it follows that 
$\lbrack T^\flat T, J_n\Omega\rbrack=0.$

Now, $C\left(J_n\Omega\right)^k=C\left(T^\flat T\right)\left(J_n\Omega\right)^k=C\left(J_n\Omega\right)^k\left(T^\flat T\right)$ which means that the minimality matrix $\mathcal{O}$ satisfies $\mathcal{O}=\mathcal{O}T^\flat T$. Because the system is minimal $\mathcal{O}$ must be full rank, hence $T^\flat T=\mathds{1}$.

Finally, it remains to show that the matrix $T$ generating the equivalence class is of the form 
\[T=\left(\begin{smallmatrix} T_1 &T_2\\T_2^{\#}&T_1^{\#}\end{smallmatrix}\right).\]
To see this, observe that $CA^k$, $C'A'^k$ must be of the of this doubled up form for $k\in\{0,1,2,\hdots\}$. Writing $CA^k$, $C'A'^k$, and $T$ as $\left(\begin{smallmatrix} P_{(k)} &Q_{(k)}\\Q_{(k)}^{\#}&P_{(k)}^{\#}\end{smallmatrix}\right)$, $\left(\begin{smallmatrix} P'_{(k)} &Q'_{(k)}\\Q_{(k)}'^{\#}&P_{(k)}'^{\#}\end{smallmatrix}\right)$ and $T=  \left(\begin{smallmatrix} T_1 &T_2\\T_3&T_4\end{smallmatrix}\right)$, and using the above result, $C'A'^k= CA^k T^\flat $, 
it follows that 
\[P_{(k)}(T_1^{\dag}-T_4^T)+Q_{(k)}(T_3^T-T_2^{\dag})=0\]
and
\[Q_{(k)}^{\#}(T_1^{\dag}-T_4^T)+P_{(k)}^{\#}(T_3^T-T_2^{\dag})=0.\]
Hence 
\[\mathcal{O}\left[\begin{smallmatrix}    T^{\dag}_1-T_4^T\\  T_3^T-T_2^{\dag}\end{smallmatrix}\right]=0\] and so using the fact that $\mathcal{O}$ is full rank gives the required result. 
\end{proof}

Therefore, without any additional information, we can at most identify the equivalence 
class of systems related by a symplectic transformation (on the system). 
Note that the above transformation of the system matrices is equivalent  to a change of co-ordinates $\breve{\mathbf{a}}\mapsto T^\flat \breve{\mathbf{a}}$ in Eq. (\ref{langevin}).

\subsection{Identification method}\label{id:method}

 Suppose that we have constructed the transfer function from the input-output data, using for instance one of the techniques of \cite{Ljung} and \footnote{Typically this can be  done by probing the system with a known input (e.g., a coherent state with a time-dependent amplitude) and performing a measurement (e.g., homodyne or heterodyne measurement) on the output field and post-processing the data (e.g., using maximum likelihood or some other classical method \cite{Ljung}).}.

Here  we a  outline a method  to construct a system realization directly from the transfer function, for a general SISO  quantum linear system. The realization is obtained indirectly by first finding a non-physical realization and then constructing a physical one from this by applying a criterion developed in \cite{Indep}.  
The construction follows similar lines to the method described in \cite{Guta2} for passive systems. 


Let  $(A_0, B_0, C_0)$ be a triple of doubled-up matrices which constitute a  minimal realization of $\Xi(s)$, i.e.,  
\begin{equation}
\label{TF1}
\Xi(s)=\mathds{1}+C_0(sI-A_0)^{-1}B_0.
\end{equation}
For example, in  Appendix \ref{APP1} such a realization is found for an $n$-mode minimal SISO system, with matrices
$(A,C)$, possessing $2n$ distinct poles each with a non-zero imaginary
part.
Any other realization of the transfer function can be generated via a similarity transformation
\begin{equation}\label{trivialtrans}
A=TA_0T^{-1}\,\,\, B=TB_0 \,\,\, C=C_0T^{-1}.
\end{equation}
The problem here  is that in general these matrices may not describe a genuine quantum system in the sense that from a given $A, B, C$ one cannot reconstruct the pair $(\Omega, C)$. Our goal is to find a special transformation $T$ mapping $(A_0, B_0, C_0)$ to a triple $(A, B, C)$ that does represent a genuine quantum system. Such triples are characterized by the following \textit{physical realizability conditions} \cite{Indep}
\begin{equation}\label{PR}
A+A^{\flat}+C^{\flat}C=0 \,\,\, \mathrm{and}\,\,\, B=-C^\flat.
\end{equation}
Therefore, substituting \eqref{trivialtrans} into the left equation of \eqref{PR} one finds 
\begin{equation}\label{lom}
\left(T^{\dag}JT\right)A_0+A_0^{\dag}\left(T^{\dag}JT\right)+C^{\dag}_0JC_0=0 ,
\end{equation}
where the matrices $J$ here are of appropriate dimensions. 

Next, because the system is assumed to be stable it follows from \cite[Lemma 3.18]{zhou} that Eq. \eqref{lom} is equivalent to 
\begin{equation}\label{solute}
T^\flat T= J \left(T^{\dag}JT\right)=   \int^{\infty}_0 J\left(C_0e^{A_0 t}\right)^{\dag}J\left(C_0e^{A_0t}\right)dt.
\end{equation}

We now need to use a result from \cite{cascade2}, which is a sort of singular value decomposition for symplectic matrices. We state the result in a slightly different way here. 
\begin{Lemma}\label{peterf}
Let $N^{2n\times 2n}$ be a complex, invertible, doubled-up matrix and let $\mathcal{N}=N^\flat N$.
\begin{enumerate}
\item[(1)] Assume that all eigenvalues of $\mathcal{N}$ are semisimple\footnote{An eigenvalue, $\lambda$ is said to be \textit{semisimple} if its geometric multiplicity equals its algebraic multiplicity. That is, the dimension of the eigenspace associated with  $\lambda$ is equal to the multiplicity of $\lambda$ in the characteristic polynomial.}. Then there exists a symplectic matrix $W$ such that $\mathcal{N}=W\hat{N}W^\flat$ where $\hat{N}=\left(\begin{smallmatrix} \hat{N}_1&\hat{N}_2\\\hat{N}_2^{\#}&\hat{N}_1^{\#}\end{smallmatrix}\right)$ with
\[\hat{N}_1=\mathrm{diag}\left(\lambda^{+}_1, ..., \lambda^{+}_{r_1}, \lambda^{-}_1, ..., \lambda^{-}_{r_2}, \mu_1 \mathds{1}_2, ..., \mu_{r_3} \mathds{1}_2\right)\]
\[\hat{N}_2=\mathrm{diag}\left(0, ..., 0, 0, ..., 0, -\nu_1 \sigma , ..., -\nu_{r_3} \sigma\right).\]
Here $\lambda_{i}^{+}>0$, $\lambda_{i}^{-} <0$ and $\lambda^{c}_i := \mu_i+i\nu_i$ (with $\mu_i, \nu_i\in\mathbb{R}$ $\nu_i >0$) are the eigenvalues of $\mathcal{N}$. The matrix $\sigma=\left(\begin{smallmatrix} 0&-i\\i&0\end{smallmatrix}\right)$ is one of the Pauli matrices and $\mathds{1}_2$ is the identity. 

\item[(2)]
There exists another symplectic matrix $V$ such that $N=V\bar{N}W^\flat$ where $\bar{N}$ is the factorization of $\hat{N}$ $\left(\hat{N}=\bar{N}^\flat\bar{N}\right)$ given by
$\bar{N}=\left(\begin{smallmatrix} \bar{N}_1&\bar{N}_2\\\bar{N}_2^{\#}&\bar{N}_1^{\#}\end{smallmatrix}\right)$ with
\[\bar{N}_1=\mathrm{diag}\left(\sqrt{\lambda^{+}_1}, \hdots, \sqrt{\lambda^{+}_{r_1}}, 0, \hdots, 0, \alpha_1 \mathds{1}_2, \hdots, \alpha_{r_3} \mathds{1}_2\right)\]
\[\bar{N}_2=\mathrm{diag}\left(0, \hdots, 0, \sqrt{|\lambda^{-}_1|}, \hdots, \sqrt{|\lambda^{-}_{r_2}|},-\beta_1 \sigma , \hdots,-\beta_{r_3} \sigma\right).\]
The coefficients $\alpha_i$ and $\beta_i$ are determined from $\mu_i$ and $\nu_i$ via
\begin{itemize}
\item[(i)] If $\mu_i\geq0$, then $\alpha_i=\sqrt{\mu_i}\mathrm{cosh}x_i$, $\beta_i=\sqrt{\mu_i}\mathrm{sinh}x_i$, with $x_i=\frac{1}{2}\mathrm{sinh}^{-1}\frac{\nu}{\mu}$.
\item[(ii)] If $\mu_i\leq0$, then $\alpha_i=\sqrt{|\mu_i|}\mathrm{sinh}x_i$, $\beta_i=\sqrt{|\mu_i|}\mathrm{cosh}x_i$, with $x_i=\frac{1}{2}\mathrm{sinh}^{-1}\frac{\nu}{|\mu|}$.
\item[(iii)] If $\mu_i=0$, then $\alpha_i=\beta_i=\sqrt{\frac{\nu_i}{2}}$.
\end{itemize}
\end{enumerate}
\end{Lemma} 
The lemma can be extended beyond the semisimple assumption, but since the latter holds for generic matrices \cite{cascade2}, it suffices for our purposes.

We can therefore use  Lemma \ref{peterf} together with Eq. \eqref{solute} in order to write the ``physical'' 
$T$ as $T=V\bar{T}W^\flat$, where $W$ and $\bar{T}$ can be computed as in the lemma above, and $V$ is a symplectic matrix. However, since the QLS equivalence classes are characterized by symplectic transformation, this means that 
$T_0=\bar{T}W^\flat$ transforms 
$(A_0, B_0, C_0)$ to the matrices of a quantum systems satisfying the realizability conditions. Finally, we can solve to find the set of physical parameters $(\Omega, C)$, which are given in terms of $(A_0, B_0, C_0)$, as
\begin{align*}
C&= C_0W\bar{T}^{-1},  \\
\Omega&=i\left(\bar{T}W^\flat A_0W\bar{T}^{-1} +\frac{1}{2}\left(\bar{T}^{\flat}\right)^{-1}W^\flat C_0^\flat C_0W\bar{T}^{-1}\right) .
\end{align*}

\begin{remark}
Note that, by assumption, $\Xi(s)$ is the  transfer function of a QLS. 
Since the original triple $(A_0, B_0, C_0)$ is minimal, this implies that there exists a nonsingular $T$ satisfying \eqref{solute}, so the right side of \eqref{solute} is nonsingular, which eventually leads to a nonsingular transformation 
$T$ computed using Lemma  \ref{peterf}.
\end{remark}


\begin{remark}
The proof also holds for multiple-input-multiple-output (MIMO) systems provided that  one can find a minimal doubled-up (non-physical) realization beforehand. 
\end{remark}

\subsection{Cascade realization of QLS}\label{seriesp}

Recently, a synthesis result has been established showing that the transfer function of a ``generic''  QLS has a pure cascade realization \cite{cascade}. Translated to our setting, this means that given a $n$-mode  QLS $(C,\Omega)$, one can construct an equivalent system (i.e., with the same transfer function) which is a series product of single mode systems. The result holds for a large class of systems characterized by the fact that the matrix $A$ admits a certain symplectic Schur decomposition, which holds for a dense, open subset of the relevant set of matrices. 


Assuming that such a cascade is possible, the transfer function is an $n$-mode product of single mode transfer functions, which are given by 
$$
\Xi_i (s)=\left(\begin{smallmatrix}\Xi_{i-}(s)&\Xi_{i+}(s)\\ &\\{\Xi_{i+}(s^{\#})}^{\#}& {\Xi_{i-}(s^{\#})}^{\#}\end{smallmatrix}\right).
$$
Further, we can stipulate that the coupling to the field is of the form $C=\Delta(C_-, 0)$, with each element of $C_-$ being real and positive. Indeed, since the system is assumed to be stable,  there exists a local symplectic transformation on each mode so that coupling is purely passive.  The point of this requirement is that it fixes all the parameters, so that under these restrictions each equivalence class from Sec. \ref{identifiability} contains exactly one element. Note that the Hamiltonian may still have both active and passive parts.  
Therefore, each one mode system in the series product is characterized by three parameters, $c_i, \Omega_{i-}\in\mathbb{R}$ with 
$c_i\neq 0$, and  $\Omega_{i+}\in\mathbb{C}$. If $\Omega_{i+}=0$ then the mode is passive. Actually, it is more convenient for us here to reparametrize the coefficients so that 
\[   \Xi_{i-}(s)= \frac{s^2-x^2_i-y^2_i +  2{ix_i\theta_i}   }{\left(s+x_i+y_i \right) \left(s+x_i-y_i\right)    },  \]
\[\Xi_{i+}(s)=\frac{-2ix_ie^{i\phi_i}\sqrt{y_i^2+\theta_i^2}}{\left(s+x_i+y_i \right) \left(s+x_i-y_i\right)    },\]
where $x_i=\frac{1}{2}c_i^2$,     $y_i= \sqrt{|\Omega_{i+}|^2-\Omega_{i-}^2}$, $\theta_i=\Omega_{i-}$ and $\phi_i=\mathrm{arg}(\Omega_{i+})$.
Therefore, from the properties of the individual $\Xi_{i\pm}(s)$, one finds that $\Xi_{-}(s)$ and $\Xi_{+}(s)$ can be written as
\begin{eqnarray}
\label{form1}
\Xi_{-}(s)&=&\prod\limits_{i=1}^n\frac{\left(s-\lambda_i\right)\left(s+\lambda_i\right)}{\left(s+x_i+y_i\right)\left(s+x_i-y_i\right)}
\\
\label{form2}
\Xi_{+}(s)&=&\gamma\frac{  \prod\limits_{i=1}^{j}     \left(s-\gamma_i\right)\left(s+\gamma_i\right)}{ \prod\limits_{i=1}^n \left(s+x_i+y_i\right)\left(s+x_i-y_i\right)},
\end{eqnarray}
with $\gamma, \gamma_i, \lambda_i\in\mathbb{C}$, $x_i\in\mathbb{R}$, and $y_i$ either real or imaginary, while $j$ is some number between $1$ and $n-1$. In particular, the poles are either in real pairs  or in complex-conjugate pairs.

Furthermore, there is a possibility that some of the poles and zeros may cancel in \eqref{form1} and \eqref{form2}, and as a result some of these poles and zeros could be fictitious (see proof of Theorem \ref{mainresult} later where this becomes important).

For passive systems such a cascade realization is always possible \cite{Indep, petersen} and  each single mode system is passive. We show how this may be done in the following example.

\begin{exmp}\label{sisoexample}
Consider a  SISO PQLS $(C, \Omega)$ and let  $z_1, z_2, \dots ,z_m$ be the eigenvalues of $A=-i\Omega-\frac{1}{2}C^{\dag}C$. Then the transfer function is given by
\begin{align*}\Xi(s)&=\frac{\mathrm{Det}(s-{A}^{\#})}{\mathrm{Det}(s-{A})}\\&=\frac{s-{z}^{\#}_1   }{s-z_1}\times \frac{s-{z}^{\#}_2   }{s-z_2}\times ...\times \frac{s-{z}^{\#}_1   }{s-z_1}.\end{align*}
Now, comparing each term in the product with the transfer function of a SISO system of one mode, i.e., 
\[\Xi(s)=\frac{s+i\Omega-\frac{1}{2}|c|^2}{s+i\Omega+\frac{1}{2}|c|^2},\]
it is clear that each represents the transfer function of a bona-fide PQLS with Hamiltonian and coupling parameters given  by $\Omega_i=-\mathrm{Im}(z_i)$ and $1/2|c_i|^2=-\mathrm{Re}(z_i)$. This realization of the transfer function is   a cascade of optical cavities.
Furthermore, we note that the order of the elements in the series product is irrelevant; in fact a differing order can be achieved by a change of basis on the system space (see Sec. \ref{identifiability}). 
\end{exmp}

In actual fact this result enables us to  find a system realization directly from the transfer function, thus offering a parallel strategy to the realization method in Sec. \ref{id:method} for passive systems. Note that a similar brute-force approach for finding a cascade realization of a general SISO system is also possible.  However, the active case is more involved than the passive case, as the transfer function is characterized by two quantities, $ \Xi_-(s)$ and $\Xi_+(s)$, rather than just one. For this reason and also that Sec. \ref{id:method} indeed already offers a viable realization anyway, we do not discuss the result here.

\section{Power spectrum system identification}
\label{powerspectrum}
Until now we addressed the system identification problem from a \emph{time-dependent} input perspective. 
We are now going to change viewpoint and consider a setting where the input fields are \emph{stationary} (quantum noise) but may have a non-trivial covariance matrix (squeezing). In this case the characterization of the equivalence classes boils down to finding which systems have the same \textit{power spectrum}, a problem which is well understood in the classical setting \cite{anders} but has not been addressed in the quantum domain.

The input state is ``squeezed quantum noise'', i.e., a zero-mean, pure Gaussian state with time-independent increments, which is completely characterized by its covariance matrix ${V}={V}(M,N)$ cf.  Eq.  \eqref{Ito}. In the frequency domain the state can be seen as a continuous tensor product over frequency modes of squeezed states  with covariance ${V}(M,N)$. Since we deal with a linear system, the input-output map consists of applying a (frequency dependent) unitary Bogoliubov transformation whose linear symplectic action on the frequency modes is given by the transfer function 
$$
\breve{\bf b}^{out} (-i\omega)  =  
\Xi (-i\omega) \breve{\bf b} (-i\omega).
$$
Consequently, the output state is a Gaussian state consisting of independent frequency modes with covariance matrix
\begin{eqnarray*} 
\left< \breve{\bf b}^{out} (-i\omega) \breve{\bf b}^{out} (-i\omega^\prime)^\dagger \right>=   
 \Psi_{V}(-i\omega) \delta(\omega-\omega^\prime),  
 \end{eqnarray*}
 where $ \Psi_{V}(-i\omega)$ is the restriction to the imaginary axis of the \emph{power spectral density} (or power spectrum) defined in the Laplace domain by
\begin{equation}\label{powers}
\Psi_{V}(s)= \Xi(s){V}\Xi(-s^{\#})^{\dag}.
\end{equation}

Our goal is to find which system parameters are identifiable in the stationary regime where the quantum input has a given covariance matrix ${V}$. Since in this case the output is uniquely defined by its power spectrum 
$\Psi_{V}(s)$ this reduces to identifying the equivalence class of systems with a given power spectrum. 
Moreover, since the power spectrum depends on the system parameters via the transfer function, it is clear that one can identify ``at most as much as'' in the time-dependent setting discussed in Sec. \ref{identifiability}. In other words  
the corresponding equivalence classes are at least as large as those described by symplectic transformations \eqref{eq.equivalene.classes}. 
%

In the analogous classical problem, the power spectrum can also be computed from the output correlations. 
The spectral factorization problem \cite{Youla} is tasked with finding a transfer function  from the power spectrum. There are known algorithms \cite{Youla, Davies} to do this. From the latter,  one then finds a system realization (i.e. matrices governing the system dynamics) for the given transfer function \cite{Ljung}. The problem is that the map from power spectrum to transfer functions is  non-unique, and each factorization could lead to system realizations of differing dimension. For this reason, the concept of \textit{global minimality} was introduced in \cite{kalman} to select the transfer function with smallest system dimension. This raises the following question: Is global minimality sufficient to uniquely identify the transfer function from the power spectrum ? 
The answer is in general negative \footnote{However, under the assumption that the transfer function be \textit{outer} the construction of the transfer function from the power spectrum is unique (see \cite{nerve}).}  , as discussed in \cite{anders, glover} (see also Lemma 2 and Corollary 1 in \cite{nerve} for a nice review). Our aim is to address these questions in the quantum case. In the following section we define an analogous notion of global minimality, and characterize globally minimal systems in terms of their stationary state. Afterwards we show that for SISO systems which admit a cascade realization the power spectrum and transfer function identification problems are equivalent. 
\subsection{Global minimality}\label{GHM}

As discussed earlier, in the time-dependent setting it is meaningful to restrict the attention to minimal systems, as they provide the lowest dimensional realizations which are consistent with a given input-output behavior. In the stationary setting however, it may happen that a minimal system can have the same power spectrum as a lower dimensional system. For instance if the input is the vacuum, and the system is passive then the stationary output is also vacuum and the power spectrum is trivial, i.e., the same as that of a zero-dimensional system. We therefore need to introduce a more restrictive minimality concept, as the stationary regime (power spectrum) 
counterpart of time-dependent (transfer function) minimality. The results of this section are valid for general MIMO systems and do not assume the existence of a cascade realization.

\begin{defn}
A system $\mathcal{G}=\left(S, C, \Omega\right)$ is said to be \textit{globally minimal} for input covariance ${V}$ if there exists no lower dimensional system with the same power spectrum $\Psi_{V}$. We call $\left(\mathcal{G}, {V}\right)$ a \textit{globally minimal pair}.
\end{defn}
%

Before stating the main result of this section we briefly review some symplectic diagonalization results which will be used in the proof. 
Consider a $k$-modes cv system with canonical coordinates $\breve{\bf c}$ and  a zero-mean Gaussian state with 
covariance matrix ${V} := \left< \breve{\bf c} \breve{\bf c}^\dagger \right>$. Any change of canonical coordinates which preserves the commutation relations is of the form $ \breve{\bf c} \mapsto  \breve{\bf c}^\prime = S  \breve{\bf c}$ where $S$ is a symplectic transformation $S$, cf. Definition \ref{def.symplectic}. In the basis  $ \breve{\bf c}^\prime$, the state has covariance matrix 
${V}^\prime = S{V}S^\dagger$. In particular there exists a symplectic transformation such that the modes ${\bf c}^\prime$ are independent of each other, and each of them is in a vacuum or a thermal state i.e. 
${V}^\prime_i := \left< \breve{c}_i^\prime \breve{c}_i^{\prime \dagger} \right>= \left(\begin{smallmatrix} n_i+1&0\\ 0&n_i\end{smallmatrix}\right)$ where $n_i$ is the mean photon number. We call $\breve{\bf c}^\prime$ a canonical basis, and the elements of the ordered sequence $n_1 \leq  \dots \leq n _k$ the \emph{symplectic eigenvalues} of ${V}$. The latter give information about the state's purity: if all $n_i=0$ the state is pure, if all $n_i>0$ the state is fully mixed. 
More generally, we can separate the pure and mixed modes and write ${\bf  c}^\prime= ({\bf c}_p^T, {\bf c}_m^T)^T$.

This procedure can be applied to the $m$ input modes ${\bf b}$, with covariance ${V}(N,M)$. Since the input is assumed to be pure, we have $S_{\mathrm{in}} {V}(N,M)S_{\mathrm{in}}^\dagger= {V}_{\mathrm{vac}}$  where $S_{\mathrm{in}}$ is a symplectic transformation and ${V}_{\mathrm{vac}}$ is the vacuum covariance matrix. The interpretation is that any pure squeezed state looks like the vacuum when an appropriate symplectic ``change of basis'' is performed on the original modes.  

Similarly, we can apply the above procedure to the stationary state of the system. Its covariance matrix $P$ is the solution of the Lyapunov equation 
 \begin{equation}\label{eq.Lyapunov}
 {A}P+P{A}^{\dag}+{C}^{\flat}{V}({C^\flat})^{\dag}=0
 \end{equation}
By  an appropriate symplectic transformation we can change to a canonical basis 
$\breve{\bf a}^\prime = S_{\mathrm{sys}} \breve{\bf a}$ such that  ${\bf a}^{\prime T} =({\bf a}_p^T , {\bf a}_m^T)$. The system matrices are now $A^\prime = S_{\mathrm{sys}}A S_{\mathrm{sys}}^\flat, C^\prime = CS_{\mathrm{sys}}^\flat$. Note that this transformation is of the  form prescribed by Theorem \ref{symplecticequivalence}, but the interpretation here is that we are dealing with the same system seen in a different basis, rather than a different system with the same transfer function.

By combining the two symplectic transformations we see that any linear system with pure input can be alternatively described as a system with vacuum input and a canonical basis of creation and annihilation operators.

%
%

The following theorem links global minimality with the purity of the stationary state of the system. 
 \begin{thm}\label{equivalence} 
 Let $\mathcal{G}:= \left(S, C,\Omega\right)$ be a QLS with \textit{pure} squeezed input of covariance 
 ${V} = {V}(M,N)$.

(1) The system  is globally minimal  if and only if the (Gaussian) stationary state with covariance $P$ satisfying the Lyapunov equation \eqref{eq.Lyapunov} is fully mixed.

(2) A non-globally minimal system is the series product of its restriction to the pure component and the mixed component. 

(3) The reduction to the mixed component is globally minimal and has the same power spectrum as the original system. 
\end{thm}



\begin{proof}


Let us prove the result first in the case $S=\mathds{1}$.

First, perform a change of system and field coordinates as described above, so that the input is in the vacuum state, while the system modes decompose into its ``pure'' and ``mixed'' parts ${\bf a}^{\prime T} =({\bf a}_p^T , {\bf a}_m^T)$. Note that this transformation will alter the  coupling and Hamiltonian matrices accordingly, but  we still denote them $\Omega$ and $C$ to simplify notations. Therefore,  in this basis the stationary state of the system is given by the covariance
$$
P=\left(\begin{smallmatrix} R+\mathds{1} & 0\\0 & R\end{smallmatrix}\right), \qquad R=\left(\begin{smallmatrix} 0&0\\0& R_{m}\end{smallmatrix}\right)
$$
and satisfies the Lyapunov equation \eqref{eq.Lyapunov}.

($\implies$)  We show that if the system has a pure component, then it is globally reducible. Let us write  $A_{\pm}$ and $C_{\pm}$ as block matrices according to the pure-mixed splitting 

\[
A_{\pm}=\left(\begin{matrix} A_{\pm}^{pp}&A_{\pm}^{pm}\\A_{\pm}^{mp}&A_{\pm}^{mm}\end{matrix}\right), \qquad
C_{\pm}=\left(C_{\pm}^p, C_{\pm}^m\right),
 \]
so that the Lyapunov equation \eqref{eq.Lyapunov} can be seen as a system of 16 block matrix equations.  
Taking the (1,1) and (1,3) blocks, which correspond to the $\left<\mathbf{a}_p\mathbf{a}_p^{\dag}\right>$ and $\left<\mathbf{a}_p\mathbf{a}_p\right>$ components of the stationary state, one obtains
\begin{align}
&A_{-}^{pp}+A_{-}^{pp \dag}+C_{-}^{p\dag}C_{-}^p=0\label{block1}\\
&A_{+}^{pp T}-C_{-}^{p \dag}C_+^p=0\label{block3}.
\end{align}
Since $A_{-}^{pp}=-i\Omega_{-}^{pp}  - 1/2(C_{-}^{p \dag} C_{-}^p - C_{+}^{p T} C_{+}^{p\#})$, Eq. (\ref{block1}) implies that $C_{+}^{pT} C_{+}^{p \#}=0$, hence $C_{+}^{p}=0$. Therefore, using this fact in Eq. (\ref{block3}) gives $A_+^{pp}=0$, hence $\Omega_{+}^{pp}=0$. 
These two tell us that the pure part contains only passive terms.

Consider now the $(1, 2)$ and $(2, 3)$ blocks, which correspond to the $\left<\mathbf{a}_p\mathbf{a}_m^{\dag}\right>$ and $\left<\mathbf{a}_m\mathbf{a}_p\right>$ components of the stationary state. From this,  we get 
\begin{align}
&A_{-}^{pm}(R_m+\mathds{1})+
A_{-}^{pm \dag}+C_{-}^{p \dag}C_{-}^m=0\label{block4}\\
&(R_m+1)A_+^{pm T}
=0\label{block5}.
\end{align}
Since $A_{-}^{pm}+A_{-}^{pm \dag}+C_{-}^{p \dag}C_{-}^m=0$, and $R_m$ is invertible, 
Eq. (\ref{block4}) implies $A_{-}^{pm}=0$. Similarly, Eq. (\ref{block5}) implies that $A_{+}^{pm}=0$.

Let $\mathcal{G}^p:= (\mathds{1}, \Omega^{pp}, C^p)$ be the system consisting of the pure modes, with 
$\Omega^{pp}= \Delta(\Omega^{pp}_{-}, 0)$ and 
$C^p = \Delta (C_-^p, 0)$. Let $\mathcal{G}^m:= (\mathds{1}, \Omega^{mm}, C^m) $ be the system consisting of the mixed modes with 
$\Omega^{mm}= \Delta(\Omega^{mm}_{-}, \Omega^{mm}_{+})$ and $C^m = \Delta (C_-^m, C_+^m)$. We can now show that the original system is the series product (concatenation) of the pure and mixed restrictions
$$
\mathcal{G} = \mathcal{G}^m \triangleleft \mathcal{G}^p.
$$
Indeed, using the fact that $C_{+}^p= \Omega_{+}^{pp}= A_{-}^{pm} = A_{+}^{pm}=0$, one can check that the series product has required matrices  \cite{g2}
$$
C_{series} = \tilde{C}^p + \tilde{C}^m = C
$$
and 
$$
\Omega_{series} = \tilde{\Omega}^{pp}+ \tilde{\Omega}^{mm} + \mathrm{Im}_\flat (\tilde{C}_m^\flat \tilde{C}_p)
$$
where the  ``tilde'' notation stands for block matrices where only one block is nonzero, e.g., $\tilde{C}^p= (C^p,  0)$, and 
$ \mathrm{Im}_\flat X:= (X- X^\flat)/2i $.

Now, let $\Xi^{p,m}(s)$ denote the transfer functions of  $\mathcal{G}^{p,m}$; since the transfer function of a series product is the product of the transfer functions, we have $\Xi(s) = \Xi^{m}(s)\cdot \Xi^{p}(s)$.  Furthermore, since  
$\mathcal{G}^{p}$ is passive and the input is vacuum, we have $\Psi^p_{ V}(s)   =  \Xi^p(s) {V} \Xi^p(-s^{\#})^\dagger = {V} $ so that 
$$
\Psi_{ V}(s) = \Xi(s) {V} \Xi(-s^{\#})^\dagger =  \Xi_m(s) {V} \Xi_m(-s^{\#})^\dagger 
$$
which means that the original system was globally reducible (not minimal).

($\impliedby$)
We now show that if the system's stationary state is fully mixed, then it is globally minimal. The key idea is that a sufficiently long block of output has a finite symplectic rank (number of modes in a mixed state in the canonical decomposition) equal to twice the dimension of the system. Therefore the dimension of a globally minimal system is ``encoded'' in the output. This is the linear dynamics analog of the fact that stationary outputs of finite dimensional systems (or translation invariant finitely correlated states) have rank equal to the square of the system dimension (or bond dimension) \cite{GUTA4}. To understand this property consider the system (S) together with the output at a long time 
$2T$, and split the output into two blocks: $A$ corresponding to an initial time interval $[0,T]$ and  $B$ corresponding to 
$[T,2T]$. If the system starts in a pure Gaussian state, then the $S+A+B$ state is also pure. By ergodicity, at time $T$ the system's state is close to the stationary state with symplectic rank $d_m$. At this point the system and output block $A$ are in a pure state so by appealing to the ``Gaussian Schmidt decomposition'' \cite{WOLF} we find that the state of the block $A$ has the same symplectic eigenvalues (and rank $d_m$) as that of the system. In the interval $[T,2T]$ the output $A$ is only shifted without changing its state, but the correlations between $A$ and $S$ decay. Therefore the joint $S+A$ state is close to a product state and has symplectic rank $2 d_m$. On the other hand we can apply the Schmidt decomposition argument to the pure bipartite system consisting of $S+A$ and $B$ to find that the symplectic rank of $B$ is $2d_m$. By ergodicity, $B$ is close to the stationary state in the limit of large times, which proves the assertion.

To extend the result to $S\neq\mathds{1}$, 
instead perform the change of field co-ordinates $V\mapsto S_{\mathrm{in}}S^bV\left(S_{\mathrm{in}}S^b\right)^{\dag}$ at the beginning. The proof for this case then follows as above because in this basis $S=\mathds{1}$.
\end{proof}
 
%


This result enables one to check global minimality by computing the symplectic eigenvalues of the stationary state. 
If all eigenvalues are nonzero, then the state is fully mixed and the system is globally minimal.  We emphasize that the argument relies crucially on the fact that the input is a pure state. For mixed input states and in particular classical inputs, the stationary state may be fully mixed while the system is non globally minimal.

The next step is to find out which parameters of a globally minimal system can be identified from the power spectrum.

\section{Comparison of power spectrum and transfer function identifiability}\label{main1}

\subsection{Power spectrum identifiability result}

The main result of this section is the following theorem which shows that two globally minimal SISO systems have the same power-spectrum if and only if they have the same transfer function, and in particular are related by a symplectic transformation as described in Theorem \ref{symplecticequivalence}.

\begin{thm}\label{mainresult}
Let $\left(C_1, \Omega_1\right)$ and $\left(C_2, \Omega_2\right)$ be two globally minimal SISO systems for fixed pure 
input with covariance ${V}(N,M)$, which are assumed to be generic in the sense of \cite{cascade}. Then
$$
\Psi_{1}(s)=\Psi_{2}(s) \,\, \mathrm{for}~\mathrm{all}~s \quad \Leftrightarrow \quad \Xi_1(s)=\Xi_2(s)\,\, \mathrm{for}~\mathrm{all}~s 
$$
\end{thm}
 \begin{proof}
Recall that the power spectrum of a system $\left(C, \Omega\right)$ is given by $\Xi(s)V\Xi(-s^{\#})^{\dag}$. 
Therefore, if $\Xi_1(s)=\Xi_2(s)$ then $\Psi_{1}(s)=\Psi_{2}(s) $. We will now prove the converse. 

Writing $V$ as $S_0\left(\begin{smallmatrix}1&0\\0&0\end{smallmatrix}\right)S_0^{\dag}$ for some symplectic matrix $S_0$,  we express the power spectrum as   $S_0\tilde{\Xi_i}(s)V_{\mathrm{vac}}\tilde{\Xi_i}(-s^{\#})^{\dag}S_0^{\dag}$, where $\tilde{\Xi}(s)$ is the transfer function of the system $\left(1,S_0^\flat C, \Omega\right)$ and $V_{\mathrm{vac}}$ is the vacuum input. 
As $S_0$ is assumed to be known, the original problem reduces to proving the same statement for systems with vacuum input. In this case the power spectrum is given by 
\begin{equation}\label{PS1}
\left(\begin{array}{cc} 
\Xi_{-}(s)  {\Xi_{-}(-s^{\#})}^{\#}      &\Xi_{-}(s)\Xi_{+}(-s)\\[2mm]
{\Xi_{+}(s^{\#})^{\#}\Xi_{-}(-s^{\#})^{\#}}&{\Xi_{+}(s^{\#})}^{\#} \Xi_+(-s)    \end{array}\right).
\end{equation}

The transfer function is completely characterized by the elements in the top row of its matrix, i.e., $\Xi_{-}(s) $ and     $\Xi_{+}(s)$. Also, $\Xi_{-}(s) $ and     $\Xi_{+}(s)$ must be of the the form \eqref{form1} and \eqref{form2}.  Our first observation is that  $\Xi_{-}(s) $ and  $\Xi_{+}(s)$ in \eqref{form1} and \eqref{form2} cannot contain poles and zeros in the following arrangement: $\Xi_{-}(s) $ has a factor like
\begin{equation}\label{topple1}
\frac{(s-{\lambda}^{\#}_i)(s+{\lambda}^{\#}_i)}{(s-{\lambda}^{\#}_i)(s-\lambda_i)}=\frac{(s+{\lambda}^{\#}_i)}{(s-\lambda_i)}
\end{equation}
 and $\Xi_{+}(s)$ contains a factor like 
\begin{equation}\label{tipple1}
\frac{(s-\lambda_i)(s+\lambda_i)}{(s-{\lambda}^{\#}_i)(s-\lambda_i)}=\frac{(s+\lambda_i)}{(s-{\lambda}^{\#}_i)}.
\end{equation}
For if this were the case and assuming that this could be done $k$ times, then our original system could be decomposed as a cascade (series product) of two systems. 
\begin{itemize}
\item[(i)]
The first system is a $k$-mode passive system with transfer function 
\begin{equation}\label{huck1}
\Xi^{(1)}(s)=\left(\begin{smallmatrix} \Xi_-^{(1)}(s)&0\\0&\Xi_-^{(1)}(s^{\#})^{\#}\end{smallmatrix}\right), 
\end{equation}
where 
$$
\Xi_-^{(1)}(s)=\prod^{k}_{i=1}\frac{(s+{\lambda}^{\#}_i)}{(s-\lambda_i)}, \qquad
\Xi_-^{(1)}(s^{\#})^{\#}=\prod^{k}_{i=1}\frac{(s+\lambda_i)}{(s-{\lambda}^{\#}_i)}.
$$
Note that by Example \ref{sisoexample} it is physical. 
\item[(ii)] The second system has $n-k$ modes and transfer function 
\begin{equation}\label{huck2}
\Xi^{(2)}(s)=   \left(\begin{smallmatrix} \Xi_-^{(2)}(s)&\Xi_+^{(2)}(s)\\\Xi_+^{(2)}(s^{\#})^{\#}&\Xi_-^{(2)}(s^{\#})^{\#}\end{smallmatrix}\right),\end{equation}
where 
$$ \Xi_-^{(2)}(s)=\Xi_-(s)   \prod^{k}_{i=1}\frac{\left(s+{\mu}^{\#}_i\right)}{\left(s-\mu_i\right)},$$
$$ \Xi_+^{(2)}(s)=    \Xi_+(s)   \prod^{k}_{i=1}\frac{\left(s+\mu_i\right)}{\left(s-{\mu}^{\#}_i\right)}.$$
It can be shown that there exists a minimal physical quantum system with this transfer function (see Appendix \ref{APP2}).
\end{itemize}
Since $\Xi^{(1)}(s)$ is passive, 
$$\Xi^{(1)}(s)V_{\mathrm{vac}}\Xi^{(1)}(-s^{\#})^{\dag}=V_{\mathrm{vac}}$$ and hence this $k$-mode system is not visible from the power spectrum, while the power spectrum is the same as that of the lower dimensional system $\Xi^{(2)}(s)$. 
Therefore we have a contradiction to global minimality.

We will now construct $\Xi_{-}(s) $ and     $\Xi_{+}(s)$ directly from the power spectrum. 
This is equivalent to identifying their poles and zeros \footnote{Note that some of the poles and zeros in \eqref{form1} and \eqref{form2} may be  ``fictitious'' and so will not be required to be identified.}. To do this we must identify all poles and zeros of $\Xi_-(s)$ and $\Xi_+(s)$ from the three quantities:
\begin{align}
& \label{blue}  \Xi_{-}(s)  {\Xi_{-}(-s^{\#})}^{\#},   \\
& \label{red} \Xi_{-}(s)\Xi_{+}(-s),\\
& \label{yellow} {\Xi_{+}(s^{\#})}^{\#} \Xi_+(-s).
\end{align}

First, all poles of $\Xi_-(s)$ and $\Xi_+(s)$ may be identified from the power spectrum. Indeed, due to stability, each 
pole in \eqref{blue}-\eqref{yellow} can be assigned unambiguously to either $\Xi_{-}(s) $ or $ \Xi_{+}(-s)$. However, cancellations between zeros and poles of the two terms in the product may lead to some transfer function poles not being identifiable, so we need to show that this is not possible.  Suppose that a pole $\lambda$ of  $\Xi_-(s)$ is not visible from the power spectrum. This implies  the following
\begin{itemize}
\item[(i)] from \eqref{blue}, $\lambda$ is a zero of $ {\Xi_{-}(-s^{\#})}^{\#}$  [equivalently $-{\lambda}^{\#}$ is a zero of  $ \Xi_{-}(s) $], and 
\item[(ii)] From \eqref{red}, $\lambda$ is a zero of $\Xi_{+}(-s)$ [equivalently $-\lambda$ is a zero of $\Xi_{+}(s)$].
\end{itemize}
We consider two separate cases: $\lambda$ nonreal or real. If $\lambda$ is nonreal then from 
 the symmetries of the poles and zeros in \eqref{form1} and \eqref{form2}, $\Xi_{-}(s)$ will contain a term like
\begin{equation}\label{topple}
\frac{(s-{\lambda}^{\#})(s+{\lambda}^{\#})}{(s-{\lambda}^{\#})(s-\lambda)}=\frac{(s+{\lambda}^{\#})}{(s-\lambda)}\end{equation} and
$\Xi_{+}(s) $ will contain a term like
\begin{equation}\label{tipple}
\frac{(s-\lambda)(s+\lambda)}{(s-{\lambda}^{\#})(s-\lambda)}=\frac{(s+\lambda)}{(s-{\lambda}^{\#})}.\end{equation}
By the argument above,  the system is not globally minimal  as there will be a mode of the system that is not visible in the power spectrum. Therefore all nonreal poles of $\Xi_-(s)$ may be identified. 
 A similar argument ensures that all poles of $\Xi_+(s)$ are visible in the power spectrum. 
 
 If $\lambda$ is real, we show that $\Xi_-(s)$ and $\Xi_+(s)$ will have terms of the form \eqref{topple} and \eqref{tipple} and the result will follow. Indeed since $\lambda$ is a pole of $\Xi_-(s)$, the denominator of \eqref{form1} must have  a second root at $\lambda$ since the first cancels with the term $(s-\lambda)$ which comes together with $(s+\lambda)$ in the numerator. But then, $\Xi_+(s)$ must also have a pole at $\lambda$ since otherwise 
 $|\Xi_-(-i\omega)|^2-|\Xi_+(-i\omega)|^2=1$ could not hold. A similar argument holds for a real pole of $\Sigma_+$.
 
 Therefore we conclude that all poles of $\Xi_\pm(s)$ can be identified from the power spectrum, and we focus next on the zeros.
 Unlike the case of poles, it is not clear whether a given zero in any of these plots belongs to the factor on the left or the factor on the right in each of these equation [i.e, to $\Xi_-(s)$ or $ {\Xi_{-}(-s^{\#})}^{\#}$ in \eqref{blue}, etc].

Since the poles of $\Xi_-(s)$ and $\Xi_+(s)$ may be different due to cancellations in \eqref{form1} and \eqref{form2}, it is convenient here to add in ``fictitious''  zeros into the plots \eqref{blue}-\eqref{yellow} so that 
 $\Xi_-(s)$ and $\Xi_+(s)$ have the same poles.  Note that these fictitious poles and zeros would have been present in \eqref{form1} and \eqref{form2} before simplification. From this point onwards, the zeros in \eqref{blue}-\eqref{yellow} will refer to this augmented list which includes the additional zeros.


 \underline{\textit{Real zeros}.}

In general the real zeros of $\Xi_-(s)$ and $\Xi_+(s)$ come in pairs $\pm \lambda$ [see Eqs. \eqref{form1},  \eqref{form2}],  unless a pole  and zero (or more than one)  cancel on the negative real line. Our task here is to distinguish these two cases from plots \eqref{blue}-\eqref{yellow}.  $\Xi_-(s)$ has either
  (i)  zeros at $\pm \lambda$, or
 (ii) a zero at $\lambda>0$ but not at $-\lambda$. 

In case (i) \eqref{blue} will have a double zero at each $\pm\lambda$, whereas in case (ii) \eqref{blue} will have a single zero at $\pm\lambda$.  We need to be careful here in discriminating cases (i) and (ii)  on the basis of the zeros of \eqref{blue}. For example, a double zero at $\lambda$ in \eqref{blue} could be a result of one case (i) or two case (ii) in $\Xi_-(s)$. More generally, we could have an $n$th order zero at $\lambda$ and as a result even more degeneracy is possible. A similar problem arises for the zeros of $\Xi_+(s)$ in \eqref{yellow}.


Our first observation here is that it is not possible for both $\Xi_-(s)$ and $\Xi_+(s)$ to have zeros at $\pm\lambda$ (taking $\lambda>0$ without loss of generality). If this were possible then by using the symplectic condition $|\Xi_-(-i\omega)|^2-|\Xi_+(-i\omega)|^2=1$ and the fact that we are assuming that  $\Xi_-(s)$ and $\Xi_+(s)$ have the same poles tells us that $\Xi_-(s)$ and $\Xi_+(s)$ must both  have had double poles at $-\lambda$. The upshot is that $\Xi_-(s)$ and $\Xi_+(s)$ will have terms of the form \eqref{topple1} and \eqref{tipple1}, which is a contradiction. 

Now, suppose \eqref{blue} has $n$ zeros at $\lambda>0$ and \eqref{yellow} has $m$ zeros at $\lambda>0$. Then we know that $\Xi_-(s)$ must have $\frac{n-p}{2}$ zeros at $-\lambda$ and $\frac{n+p}{2}$ zeros at $\lambda$. Also, $\Xi_+(s)$ must have $\frac{m-q}{2}$ zeros at $-\lambda$ and $\frac{m+q}{2}$ zeros at $\lambda$. The goal here is to find $p$ and $q$ because if these are known then it is clear that there must be    $\frac{n-p}{2}$ ($\frac{m-q}{2}$) type (i) zeros and $p$ ($q$) type (ii) zeros in $\Xi_-(s)$ ($\Xi_+(s)$).

By the observation above it is clear that either $p=n$ or $q=m$. Also, in \eqref{red} there will be $\frac{n+m+p-q}{2}$ zeros at $\lambda$ and $\frac{n+m+q-p}{2}$ zeros at $-\lambda$. Hence $q-p$ is known at this stage. Finally, it is fairly easy to convince ourselves that if $p=n$ but one concludes that $q=m$ (or vice versa) and using the value of $q-p$ leads to a contradiction. Hence $p$ and $q$ can be determined uniquely. For example, if $n=2$, $m=5$, $q=2$ and $p=3$ so that $q=n$ and $q-p=-1$. Then assuming wrongly that $p=5$ and using $q-p=-1$ it follows that $q=4$ and so $n$ must be 6, which is incorrect.  

 Having successfully identified all real zeros,  we now show how to identify the zeros of $\Xi_-(s)$ and $\Xi_+(s)$ away from the real axis.

\underline{\textit{Complex (nonreal) zeros}.}

Comparing the zeros of \eqref{blue} with those of \eqref{red} we find two cases in which the zeros can be assigned directly.
\begin{itemize}
\item[(i)] Case 1: Let $z$ be a zero of \eqref{blue} that is not a zero of \eqref{red}. Then $z$ must be a zero of  ${\Xi_{-}(-s^{\#})}^{\#}$. Hence $-{z}^{\#}$ is a zero of $\Xi_-(s)$.
\item[(ii)] Case 2: Let $w$ be a zero of \eqref{red} that is not a zero of \eqref{blue}. Then $w$ must be a zero of  ${\Xi_{+}(-s)}^{\#}$. Hence $-w$ is a zero of $\Xi_+(s)$.
\end{itemize}
The question now is whether this procedure enables one to identify all zeros. Suppose that there is a zero $v$ that is common to both of these plots. Then $-{v}^{\#}$ must also be a zero of \eqref{blue}. Now, if $-{v}^{\#}$ is not a zero of \eqref{red} then $v$ is identifiable  as belonging to $\Xi_-(s)$.

Therefore we can restrict our attention to the case that the zero pair $\{v,-{v}^{\#}\}$ is common to both plots. Note that in this instance the list of zeros of \eqref{yellow} will also contain $\{v,-{v}^{\#}\}$.  
Assume without loss of generality that $v$ is in the right half complex plane. 
Note that there cannot be a second zero pair $\{u,-{u}^{\#}\}$ such that $u={v}^{\#}$. If this were the case then either $\{v, -v\}$ will be zeros of $\Xi_-(s)$ and  $\{-{v}^{\#}, {v}^{\#}\}$ will  be zeros of   $\Xi_+(s)$, or  $\{u, -u\}$ will be zeros of $\Xi_-(s)$ and  $\{-{u}^{\#}, {u}^{\#}\}$ will  be zeros of   $\Xi_+(s)$. In either case by using the condition $|\Xi_-(-i\omega)|^2-|\Xi_+(-i\omega)|^2=1$ for all $\omega$ and the fact that   $\Xi_-(s)$ and $\Xi_+(s)$ have the same poles by assumption, it follows that $\Xi_-(s)$ and $\Xi_+(s)$ will have terms of the form \eqref{topple1} and \eqref{tipple1}, which contradicts global minimality. Finally, under the assumptions that the zero pair $\{v,-{v}^{\#}\}$ is common to both \eqref{red} and \eqref{blue} with no second pair at $\{u,-{u}^{\#}\}$ such that $u={v}^{\#}$, then we can conclude that $v$ must be a zero of $\Xi_-(s)$. For if this were not the case and so $-{v}^{\#}$ were a zero of $\Xi_-(s)$ then there must be another zero of $\Xi_-(s)$ at ${v}^{\#}$ (since pole-zero cancellation cannot occur in the right-half plane). Also from \eqref{red} this would require that $\Xi_+(s)$ has a zero at $-v$ (hence also $v$). Therefore we have a contradiction to the fact that there is no second pair at $\{u,-{u}^{\#}\}$ such that $u={v}^{\#}$. 


Therefore we have  successfully identified all  zeros of the transfer function away from the real axis, which completes the proof. 
\end{proof}

The theorem says that if a SISO system is globally minimal then the power spectrum is as informative as the transfer function. The result also gives a constructive method to check global minimality. Further it   enables one to  construct the transfer function of the system's globally minimal part. From this, one can then  construct a system realization of this globally minimal restriction, using the results from Sec. \ref{id:method}. We call this realization method  \textit{indirect} because one first finds a transfer function fitting the power spectrum before constructing the system realization.

\begin{Corollary} 
Let   $(C, \Omega)$ be a SISO QLS with pure  input ${V}(N,M)$. Then one can construct a globally minimal realization, $(C', \Omega')$ {indirectly} from the power  spectrum generated by the QLS $(C, \Omega)$. The realization  $(C', \Omega')$ will be  unique up to the symplectic equivalence in Theorem \ref{symplecticequivalence}.
\end{Corollary}

Note that the work here also extends a result in \cite{Naoki}. There, conditions were derived to determine when the stationary state of the linear system is pure. Here, by means of the previous theorem, we have established a test to determine if there is a \textbf{subsystem}  with a pure stationary state. 

\begin{remark}
For general input $V=S_0V_{\mathrm{vac}}S_0^{\dag}$, clearly systems of the form $\left( S_0\Delta\left(C_-,0\right), \Delta\left(\Omega_-, 0\right)\right)$ have trivial power spectrum. 
Theorem \ref{mainresult} says that these are the only such systems (up to symplectic equivalence in Theorem \ref{symplecticequivalence}). 
\end{remark}

\begin{remark}
We have assumed that the scattering or squeezing matrix, $S$, for a system is the identity in this result. In fact the scattering or squeezing  matrix is not always identifiable from the power spectrum. For example, a zero mode system with a single scattering term $S=\left(\begin{smallmatrix}e^{i\pi}&0\\0&e^{-i\pi}\end{smallmatrix}\right)$ will have trivial power spectrum.
\end{remark} 

\subsection{Power spectrum identification of passive QLSs}

In this section we consider the special case of a minimal \emph{passive} SISO QLSs. As noted before, we can therefore drop the doubled-up notation, cf. Eqs \eqref{iolp} and \eqref{eq.transfer.function}. For simplicity we will denoted $C:= C_-$, $\Omega:= \Omega_-$, and choose $S= \mathds{1}$ so that the transfer function is 
$$ 
\Xi(s)=1 -C (s\mathds{1}_n-A )^{-1}C^{\dag} = 
\frac{\mathrm{det}\left(s\mathds{1}_n + {A}^{\#} \right)}{\mathrm{det}\left(s\mathds{1}_n-A \right)}
$$
where $A= -i\Omega -\frac{1}{2}C^\dagger C$ and its spectrum is $\sigma(A) :=\{ \lambda_1,\dots ,\lambda_n\}$. 
The transfer function is a \textit{monic} rational function in $s$, with poles $p_i= \lambda_i$ in the left half plane, and zeros  
$z_i= -{p_i}^{\#}= -{\lambda^{\#}_i}$ in the right half plane.


If the input state is vacuum then the power spectrum is trivial ($\Phi_{V}=V$) and the only globally minimal systems are the trivial ones (zero internal modes).   For this reason we restrict our attention to \emph{squeezed} inputs, i.e., $M\neq 0$ in the input covariance.

\begin{thm}\label{sisogm}
Consider  a general SISO PQLS $\mathcal{G}=(C,\Omega)$  with pure input ${V}(N,M)$, such that $M\neq 0$. 

%

(1) The following are equivalent:
 \begin{itemize}
 \item[(i)] the system is globally minimal;
 \item[(ii)]  the stationary state of the system is fully mixed;
 \item[(iii)] $A$ and $A^\dagger$ have different spectra, i.e., $\sigma(A) \cap \sigma (A^\dagger) =\emptyset $;
 \item[(iv)] $A$ does not have real, or pairs of complex conjugate eigenvalues.
 \end{itemize}

(2) Let $\mathcal{P}$ be the set of all eigenvalues of $A$ that are either real or come in complex-conjugate pairs. A globally minimal realization of the system is given by the series product of one mode systems 
$\mathcal{G}_{m,i}= ( c_i=\sqrt{2|\mathrm{Re} \lambda_i}|, \Omega_i=-\mathrm{Im}\lambda_i )$ for indices $i$ such that 
$\lambda_i\notin \mathcal{P}$. 

%

\end{thm}

\begin{proof}

(1) For passive SISO systems the only nontrivial contribution to the power spectrum is from off-diagonal element, 
\begin{eqnarray*}
\Xi(s)\Xi(-s)&=&\frac{\mathrm{det}\left(s\mathds{1}_n + A^\dagger\right)}{\mathrm{det}\left(s\mathds{1}_n-A\right)}
\frac{\mathrm{det}\left(s\mathds{1}_n- A^\dagger\right)}
{\mathrm{det}\left(s\mathds{1}_n+A\right)}\\
&=& \prod_{i=1}^n \frac{s+ \lambda^{\#}_i}{s-\lambda_i}  \frac{s- \lambda^{\#}_i}{s+\lambda_i}.
\end{eqnarray*}

In the above expression, zero-pole cancellations occur if and only if $\sigma(A) \cap \sigma (A^\dagger) \neq \emptyset $, or equivalently if $A$ has a real eigenvalue or a pair of complex conjugate eigenvalues.

If no zero-pole cancellations occur, then $\sigma(A)$ can be identified from $\Xi(s)\Xi(-s)$ and the transfer function can be reconstructed. In this case the system is globally minimal. 

If cancellations do occur then this happens in one of the two types of situations: 

(a) real eigenvalue: if $\lambda_i \in \mathbb{R}$ then the corresponding term in the above product cancels

(b) complex conjugate pairs: if $\lambda_i = {\lambda}^{\#}_j$ then the $i$ and $j$ terms in the product cancel against each other. 

In both cases, the remaining power spectrum has the same form, and can be seen as the power spectrum of a series product of one-dimensional passive systems, with dimension smaller than $n$, and therefore the system is not minimal.  

This shows the equivalence of (i), (iii) and (iv) while the equivalence of (i) and (ii) follows from Theorem \ref{equivalence}.


(2) The discussion so far shows that the transfer function factorizes as the product $\Xi(s)=\Xi_{\mathrm{m}}(s)\Xi_{\mathrm{p}}(s)$ of a part corresponding to eigenvalues $\lambda_i\in\mathcal{P}$, which has a trivial power spectrum due to zero-pole cancellations, and the part corresponding to the complement which does not exhibit any cancellations. A system with transfer function $\Xi(s)$ can be realized as series product $\mathcal{G}_{m}\triangleleft\mathcal{G}_{p}$ of two separate passive systems 
with  transfer functions $\Xi_{\mathrm{m}}(s)$ and $\Xi_{\mathrm{p}}(s)$. 
As argued before, $\mathcal{G}_{p}$ has a pure stationary state which is uncorrelated to $\mathcal{G}_{m}$ or the output, while 
$\mathcal{G}_{m}$ has a fully mixed state which is correlated to the output.

Since $\mathcal{G}_{p}$ does not contribute to the power spectrum, a globally minimal realization is provided by $\mathcal{G}_{m}$,
\begin{equation}\label{partm}
\Xi_{m}(s)= \prod_{i \notin \mathcal{P}}  \frac{s+ \lambda^{\#}_i}{s-\lambda_i}
\end{equation}
Each fraction in \eqref{partm} represents a bona fide PQLS $\mathcal{G}_{m,i}$ with Hamiltonian  and coupling parameters 
$\Omega_i=-\mathrm{Im}\lambda_i $ and $1/2 |c_i|^2=-\mathrm{Re} \lambda_i $.
\end{proof}


For PQLSs we now see that it is possible to construct a globally minimal realization of the PQLS \textit{directly} from the power spectrum. Moreover, global minimality of PQLSs  may be completely understood in terms of the spectrum of the system matrix $A$, just as was the case for minimality, stability, observability, and controllability \cite{Guta2, Indep}. 
An immediate corollary of this is the following.
\begin{Corollary}
A SISO  PQLS $\mathcal{G}=(C,\Omega)$, with pure input ${V}(N, M)$ has a pure stationary state if and only if either of the following holds:
\begin{enumerate}
\item[(1)] the input is vacuum;
\item[(2)] the eigenvalues of $A$ are real or come in complex-conjugate pairs. 
\end{enumerate}
\end{Corollary}



From Theorem \ref{sisogm} there are two types of ``elementary'' systems that are not identifiable from the power spectrum for 
arbitrary input ${V}(N,M)$. Written in the doubled-up notation, these are either:
 (i) one mode systems of the form $\mathcal{G}_1=\left(\Delta(c,0),0\right)$, or 
(ii) two mode systems of the form \\$\mathcal{G}_2=\left(\Delta(c,0),\Delta(\Omega_-,0)\right)\triangleleft\left(\Delta(c,0),\Delta(-\Omega_-,0)\right)$.  
Either way it is not immediately obvious whether these systems are consistent with the nonidentifiable systems in Theorem \ref{mainresult}.  As an example we will show that this is indeed the case in the case of 
$\mathcal{G}_1$ ($\mathcal{G}_2$ is similar). 
\begin{exmp}
Consider system $\mathcal{G}_1$ for input $V(N, M)$, which is known to have (trivial) power spectrum $V(N, M)$. Therefore, in the vacuum basis of the field the system will be 
\begin{equation}\label{none}
\tilde{\mathcal{G}}_1=\left(S^{\flat}_{\mathrm{in}}\Delta(c,0),0\right)
\end{equation}
 (see Sec. \ref{GHM}) and the power spectrum will be vacuum. 
Now, as $S_0\Delta\left(c_-,0\right)=\Delta\left(c_-,0\right)S_0$ it follows that $\tilde{\mathcal{G}}_1$ must be transfer function equivalent to the system $\left(\Delta(c,0),0\right)$ in the vacuum basis. Therefore, because this system is passive we have consistency with  Theorem \eqref{mainresult}.

In fact we can even see that \eqref{none} is passive by directly computing its transfer function. One can check that
\[\Xi_-(s)=\frac{s-|c|^2/2}{s+|c|^2/2} \,\,\,\mathrm{and}\,\,\,\Xi_+(s)=0.\]
\end{exmp}

Finally, it seems that the assumption of global minimality seems to be not very restrictive; we illustrate this in the form of an example.

\begin{exmp}
Consider the following SISO PQLS with two internal modes: \[\mathcal{G}=\left((0, 2\sqrt{2}), \frac{1}{2}\left(\begin{smallmatrix} 4+x &4-x\\4-x&4+x\end{smallmatrix}\right)\right),\]
where $x\in\mathbb{R}$. We examine for which values of $x$ the system is globally minimal for squeezed inputs. One can first check that the system is minimal if and only if $x\neq 4$. In Fig. \ref {gm2} we plot the imaginary parts of the eigenvalues of $A$ and $A^\dagger$.  
\begin{figure}
\centering
\includegraphics[scale=0.25]{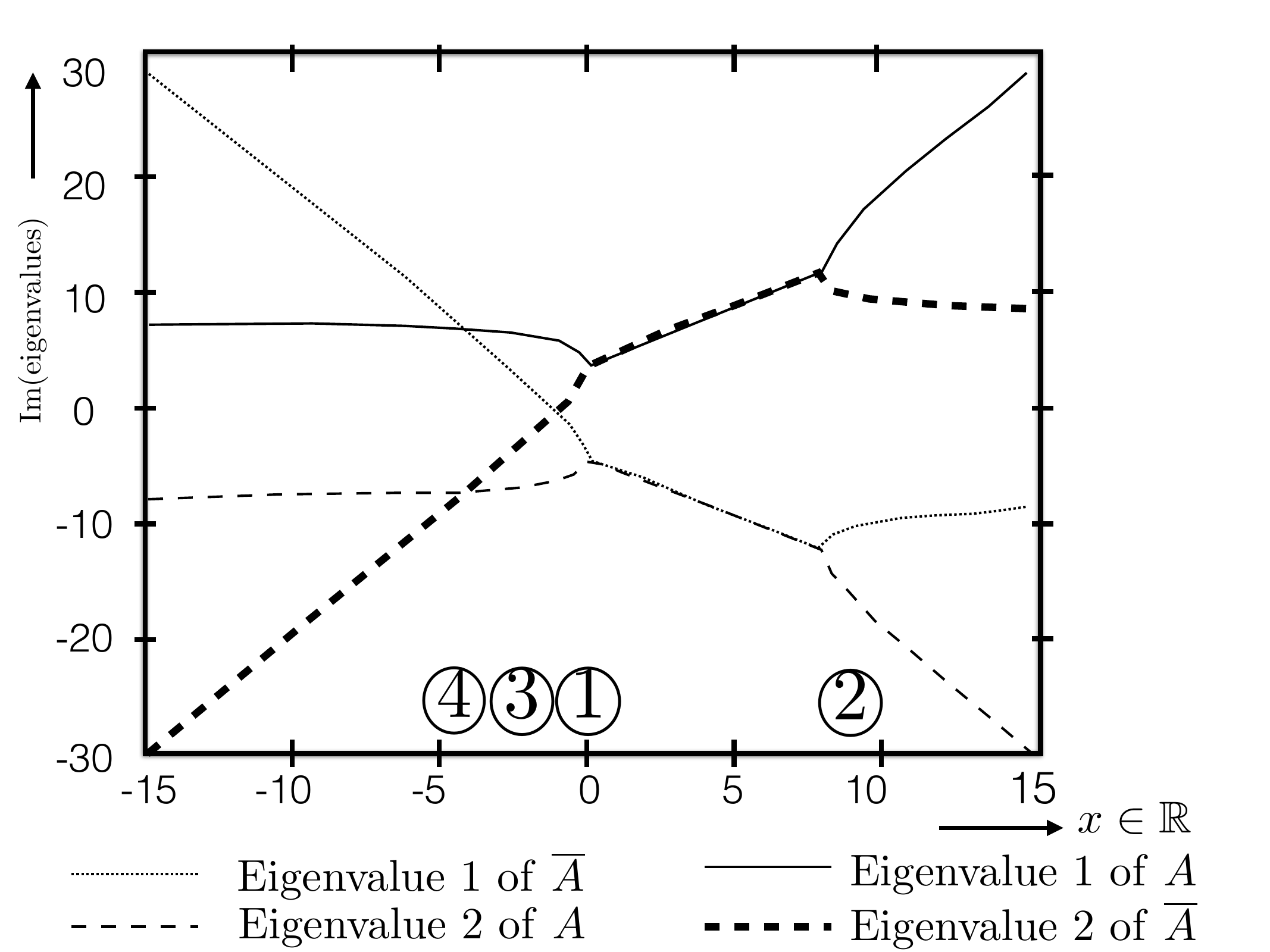}
\caption{Eigenvalues of $A$ and $A^\dagger$ as function of $x$. \label{gm2}}
\end{figure}
By Theorem \ref{sisogm}, the system is not globally minimal if any of the lines representing the eigenvalues of $A$ intersect those of 
$A^\dagger$. 
There are four points of interest that have been highlighted in the figure: 
\begin{itemize}
\item[\textcircled{1}] $x=0$: crossing of eigenvalues of $A$ but not with eigenvalues of $A^\dagger$; system is globally minimal.
\item[\textcircled{2}] $x=8$: crossing of eigenvalues $A$ but not with eigenvalues of $A^\dagger$; system is globally minimal.
\item[\textcircled{3}] $x=-1$: An eigenvalue of $A$ coincides with one of $A^\dagger$, therefore the dimension of the pure component is 1. This occurs when one eigenvalue is real.
\item[\textcircled{4}] $x=-4$: Both eigenvalues of  $A$ coincide with those of $A^\dagger$, and form a complex-conjugate pair,  
therefore the dimension of the pure space is 2.
\end{itemize}
In summary, there were only two values of $x$ for which the system is not globally minimal.
\end{exmp}
%

\subsection{Global minimality with entangled inputs}\label{newinput}
Here we show that  using an additional ancillary channel with an appropriate design of input makes it  possible to identify the transfer function from the power spectrum for \emph{all} minimal systems.

Consider the setup in Fig. \ref{entangledr}, where a pure entangled input state is fed into a SISO QLS and an additional ancillary channel. The $2\times 2$ blocks of the input ${V}(N,M)$ are
\[N=\left(\begin{smallmatrix} N_1&N_2\\{N_2}^{\#}&N_3\end{smallmatrix}\right)\,\,\, M=\left(\begin{smallmatrix} M_1&M_2\\M_2&M_3\end{smallmatrix}\right).\] The doubled-up transfer function is given by
\begin{equation}\label{fud}
\Xi(s)=\left(\begin{smallmatrix}\Xi_{-}(s)&0&\Xi_{+}(s)&0\\0&1&0&0\\{\Xi_{+}(s^{\#})}^{\#}&0&{\Xi_{-}(s^{\#})}^{\#}&0\\0&0&0&1\end{smallmatrix}\right).
\end{equation}
Now calculating the $(2,1)$ and $(1,4)$ entries of the power spectrum using (\ref{powers}), we obtain: 
$$
N_2{\Xi_{-}(s)}^{\#}+M_2{\Xi_{+}(s)}^{\#}
$$ 
and 
$$M_2\Xi_{-}(s)+N_2\Xi_{+}(s).$$
 Equivalently we may write these in matrix form as 
\[\left(\begin{smallmatrix}{N_2}^{\#}& {M_2}^{\#}\\ M_2& N_2\end{smallmatrix}\right)\left(\begin{smallmatrix}\Xi_{-}(s)\\\Xi_{+}(s)\end{smallmatrix}\right).\]
Hence if  we choose $|N_2|\neq|M_2|$ we may identify the transfer function of our SISO system uniquely. For example, such a choice of input would be $N=x\mathds{1}$ and $M=\left(\begin{smallmatrix}0&y\\y&0\end{smallmatrix}\right)$ with $x(x+1)=|y|^2$ (the purity assumption). As one can see there are no requirements on the actual QLS other than minimality. Note that in the case of passive systems we need only that $N_2$ or $M_2$ be different from zero.

\begin{figure}
\centering
\includegraphics[scale=0.35]{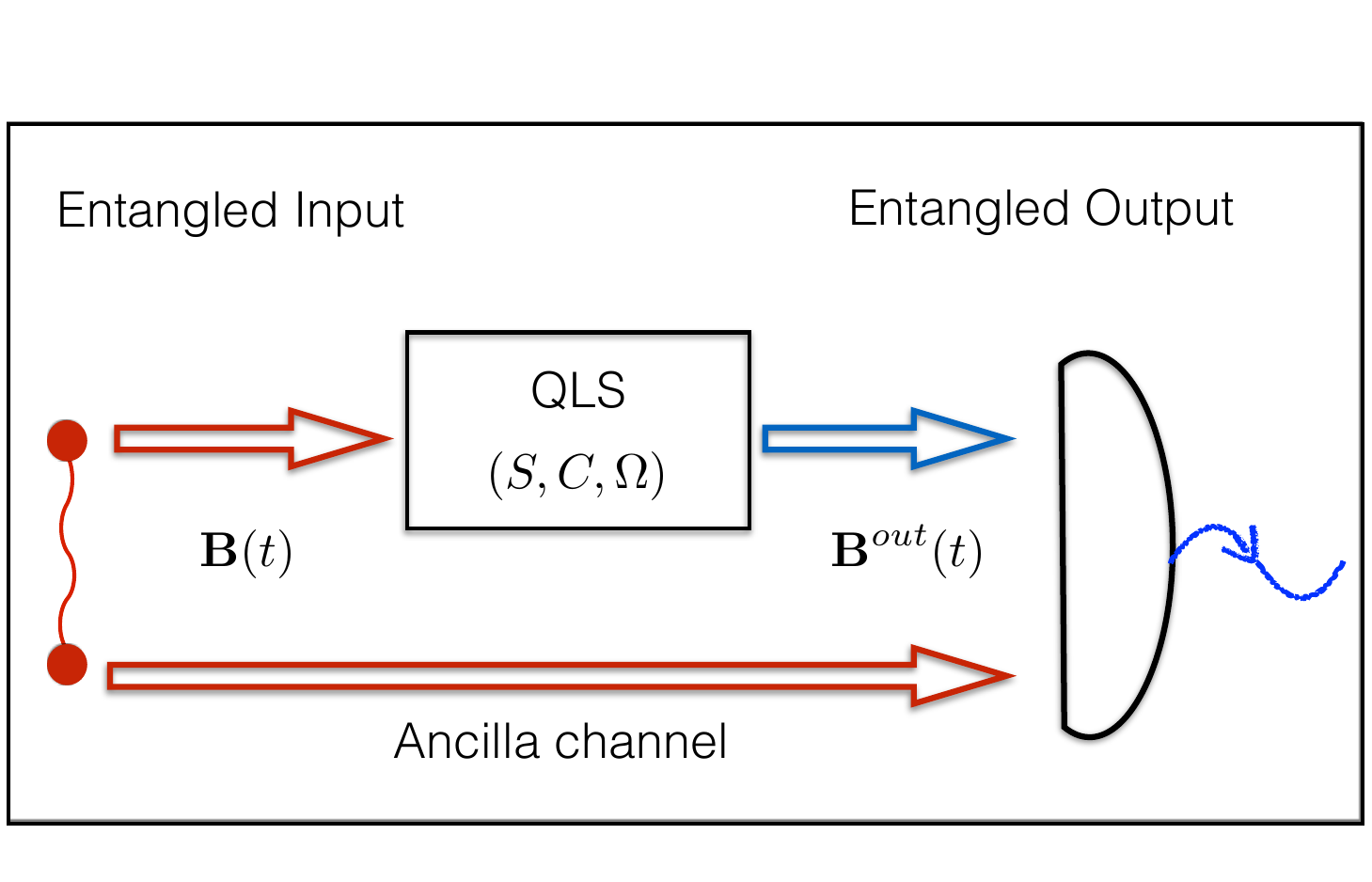}
\caption{Entangled setup discussed in Sec. \ref{newinput}. There are two channels, which are our PQLS and an additional ancilla channel. Inputs are entangled over the two channels.  \label{entangledr}}
\end{figure}

\begin{remark}
Recall from the previous subsections that the maximum amount of information we may
obtain about a PQLS from the power spectrum without the use of ancilla is that of the restriction to its globally
minimal subspace.
However, we have seen here that it is possible to construct a globally minimal pair, and hence obtain the whole transfer function simply by embedding the system in a larger space. To be clear here, there is no
contradiction because the transfer function we are attempting to identify is the one in Eq. (\ref{fud}) rather than the SISO system $\Xi(s)$. 
\end{remark}


\section{Conclusion}

We have considered the identifiability of linear system using two contrasting approaches: (1) Time-dependent input (or transfer function) identifiability and (2) stationary inputs (or power spectrum) identifiability. 
In the time-dependent approach we  characterized the equivalence class of systems with the same input-output data in Theorem \ref{symplecticequivalence}, thus generalizing the results of \cite{Guta2} to active systems. We then outlined a method to construct a (minimal and physical) realization of the system from the transfer function. In fact, all results here hold for MIMO systems.  
In the stationary input regime, Theorem \ref{equivalence} showed that global minimality is equivalent to the stationary state of the system being fully mixed.
Moreover, for a fixed pure input  generically the transfer function may be constructed uniquely from the power spectrum under global minimality. A method was also given for how to do this in Theorem \ref{mainresult}.  
Restricting to passive systems we saw that global minimality can be completely  understood simply by considering the  system matrix, $A$. In particular,  the transfer function can be constructed uniquely  from the power spectrum if and only if  none of the eigenvalues of $A$ are real nor come in complex-conjugate pairs (assuming that the input is squeezed). Finally, by using an ancillary channel it was shown that it is possible  to identify any QLS uniquely from the transfer function. 

There are several directions to extend this work. First, it is expected that all results found for the stationary input approach  can also be extended to  (i) MIMO systems and (ii) those systems beyond the generic ones considered within this paper. We intend to  address this in a future publication.    
 Given that we  now understand what is identifiable, the next step is to understand how well parameters can be estimated.   In the time-dependent approach this has been done for passive systems in \cite{Guta2, Levitt} but no such work exists for active systems or in the stationary approach at all. 
Last, it would be interesting to consider these identifiability problems in the more realistic scenario of noisy QLSs. In a QLS noise may be modelled by the inclusion of additional channels that  cannot be monitored. Understanding what can be identified here will likely be far more challenging.

\appendix

\section{Finding a minimal classical realization}\label{APP1}

In this appendix a set of (nonphysical) minimal and doubled-up matrices $(A_0, B_0, C_0)$ are found that realizes the transfer function \eqref{TF1}, which describes a (minimal) physical system $(A, C)$. 

We assume that the matrix  $A$ for   the $n$-mode minimal system, $(A, C)$, possesses  $2n$ distinct eigenvalues each with a nonzero imaginary part. This requirement  can be seen to be generic in the space of all quantum systems \cite{cascade}. Moreover, it can also be shown that if  $\lambda_i$ is a complex eigenvalue of $A$ with right eigenvector $\left(\begin{smallmatrix}R_i\\S_i\end{smallmatrix}\right)$ and left eigenvector  $\left(U_i, V_i\right)$, then ${\lambda}^{\#}_i$ is also an eigenvalue with right eigenvector $\left(\begin{smallmatrix}{S}^{\#}_i\\{R}^{\#}_i\end{smallmatrix}\right)=\Sigma{\left(\begin{smallmatrix}R^{\#}_i\\S^{\#}_i\end{smallmatrix}\right)}$ and left eigenvector $   \left({V}^{\#}_i, {U}^{\#}_i\right)= {\left(U^{\#}_i, V^{\#}_i\right)}\Sigma_n$, where     $R_i, S_i\in\mathbb{C}^{1\times n}$, $U_i, V_i\in\mathbb{C}^{n\times1}$ and $\Sigma_n:=\left(\begin{smallmatrix} 0_n&1_n\\1_n&0_n\end{smallmatrix}\right)$. That is, for each eigenvalue and eigenvector, there exists a corresponding mirror pair. 
This property follows from the fact that $A$ has the doubled-up form  $A:=\Delta\left(A_{-}, A_{+}\right)$.

We now construct a minimal realization called Gilbert's realization \cite{zhou}. 
The only thing that we need to take care of is that the realization we obtain is of the doubled-up form. 

As the transfer function may be written as 
\[\Xi(s)=\frac{N(s)}{\prod_{i=1}^n (s-\lambda_i)(s+\lambda_i)}.\]
we can perform a partial fraction expansion, so that 
\[\Xi(s)=\mathds{1}+\sum_{i=1}^n \frac{P_i}{(s-\lambda_i)} +\frac{Q_i}{\left(s-{\lambda}^{\#}_i\right)}.\]
As we show below, the matrices $P_i, Q_i$ are rank 1. Therefore there exist matrices $B_i\in\mathbb{C}^{1\times 2}$, $B'_i\in\mathbb{C}^{1\times 2}$, $C_i\in\mathbb{C}^{2\times 1}$, and $C'_i\in\mathbb{C}^{2\times 1}$ such that 
\begin{equation*}
C_iB_i=P_i \,\,\,\mathrm{and}\,\,\, C'_iB'_i=Q_i.
\end{equation*}
The Gilbert realization $A_0, B_0, C_0$ is 
\begin{equation*}
A_0:=\mathrm{diag}\left(\lambda_1, \hdots, \lambda_n,{\lambda}^{\#}_1,\hdots, {\lambda}^{\#}_n\right),
\end{equation*}
\begin{equation*}
B_0:= \left[
  \begin{array}{c}
    B_1 
    \\
     \vdots              \\
    B_n 
    \\
     B'_1 
     \\
              \vdots   
               \\
    B'_n 
  \end{array}
\right]
\end{equation*}
and
\begin{equation*}
C_0:= 
\left[
  \begin{array}{cccccc}
    C_{1}    &      \ldots & C_{n} &C'_1&\ldots&C'_n   \\
  \end{array}
\right].
\end{equation*}
From the expression of the physical transfer function we have
\begin{equation*}
C\left(s-A\right)^{-1}C^{\flat}=
\sum^n_{i=1} \frac{W_i}{s-\lambda_i}+
\frac{\Sigma{W}^{\#}_i \Sigma}{s-{\lambda}^{\#}_i }
\end{equation*}
where $W_i$ are the rank-one matrices
\begin{equation*}
W_i=\left(\begin{smallmatrix}C_-R_i+C_+S_i\\{C}^{\#}_+R_i+{C}^{\#}_iS_i\end{smallmatrix}\right) \left(\begin{smallmatrix}U_iC_-^{\dag}-V_iC_+^{\dag}&U_iC_+^T+V_iC^T_-\end{smallmatrix}\right).
\end{equation*}
Having fixed $B_i$ and $C_i$ the matrices $B'_i$ and $C'_i$ can then be chosen as 
\begin{equation}
B'_i={B}^{\#}_i\Sigma_2 \,\,\, \mathrm{and}\,\,\,
C'_i=\Sigma_2{C}^{\#}_i
\end{equation}
and so the matrices 
$(A_0, B_0, C_0)$
are of the doubled-up type.

Note that using Gilbert's realization on MIMO systems can also be seen to give a minimal doubled-up realization, but we do not discuss this any further here.

\section{Proving that there exists a minimal physical system with transfer function \eqref{huck2}}\label{APP2}
First, since we know that the system described by $\Xi(s)$ is physical, then the result of connecting it in series to another physical quantum system will be physical. To this end, consider the system 
$$\tilde{\mathcal{G}}=\mathcal{G}\triangleleft \mathcal{G}_n\triangleleft\hdots\triangleleft\mathcal{G}_1,$$
where $G$ was our original system and $G_i$ is a single mode active system with coupling $c_-=0$, $c_+=\sqrt{2\mathrm{Re}\mu_i}$, and Hamiltonian $\Omega_-=\mathrm{Im}\mu_i$, $\Omega_+=0$, where $\mu_i$ are given in the form of $\Xi^{(1)}(s)$.   Then $\tilde{\mathcal{G}}$ is physical and is described by the transfer function $\Xi^{(2)}(s)$. Also it must be stable because the transfer functions  $\Xi^{}(s)$ and $\Xi^{(1)}(s)$ have poles in the left half of the complex plane only. 
However, it is not minimal. 

To find a minimal system employ the quantum Kalman decomposition from \cite{new}. The result is that this system may be written in the form of Eqs. (103) and (104) in  \cite{new}. Hence the system is transfer function equivalent to the \textit{minimal} system  with matrices  (in quadrature form) $\left(\tilde{A}_{co}, B_{co}, C_{co}\right)$ from \cite{new}. This system gives a minimal realization of the transfer function $\Xi^{(2)}(s)$.
 It  can also can be verified that it is physical (this either  follows  because its transfer function is doubled-up and symplectic  \cite{peter} or alternatively from the results in \cite{new}) and that the matrices $\left(\tilde{A}_{co}, B_{co}, C_{co}\right)$ are of the doubled-up type, as required.
 
 Finally, since  two stable and minimal quantum systems connected in series is always minimal (see a proof of this below), then it is clear that $\Xi^{(2)}(s)$ must necessarily be of size  $n-k$. To see the previous claim,  suppose that we have two minimal systems $(C_1, A_1)$ and $(C_2, A_2)$, where $C_i$ is the coupling matrix of the system and $A_1$ is the usual system matrix. Connecting these systems in series [$(C_1, A_1)$ into $(C_2, A_2)$] we get the resultant coupling and system matrices  \cite{zhou}
$$
(C, A):=\mathbf{\Big(}\left(\begin{smallmatrix} C_1&C_2\end{smallmatrix}\right) ,\left(\begin{smallmatrix} A_1&0\\-C_2^{\flat} C_1&A_2\end{smallmatrix}\right)\mathbf{\Big)}.
$$

Recall that in order to show that the QLS (C,A) is minimal it is enough to show that the pair $(A, -C^\flat)$ is controllable \cite{Indep}.  This is equivalent to the condition that for all eigenvalues and left eigenvectors of 
$A$, i.e. $vA=v\lambda $ then $vC^\flat \neq0$ \cite{zhou}. 

 First, 
$\left(y_1, y_2 \right)A=\left(y_1, y_2 \right)\lambda $ implies $y_2A_2=y_2\lambda$. Note that by stability $\mathrm{Re}(\lambda)<0$. Hence by controllability of the second system $y_2C_2^\flat\neq0$. Suppose to the contrary that $(A, -C^\flat)$ is not controllable. Then $y_1C^\flat_1+y_2C^\flat_2=0$, which together with $\left(y_1, y_2 \right) A=\left(y_1, y_2 \right)\lambda $ would imply that 
\begin{equation}\label{apeq}
y_1\left(A_1+C^{\flat}_1C_1\right)=y_1\lambda.
\end{equation}
Since $A_1=-iJ\Omega_1-\frac{1}{2}C_1^{\flat}C_1$ then for \eqref{apeq} it is required that  $\mathrm{Re}(\lambda)>0$, which is a contradiction.

\bibliographystyle{apsrev4-1}
\bibliography{references.bib} 

\end{document}